\documentclass[12pt]{article}
\usepackage[margin=1.2in]{geometry}
\usepackage[utf8]{inputenc}
\usepackage{amsthm,amsmath,physics,mathtools,amssymb}

\usepackage{charter}

\usepackage{CJKutf8}

\usepackage{enumitem}

\usepackage{hyperref}
\hypersetup{
    colorlinks=true,
    linkcolor=blue,
    filecolor=magenta,      
    urlcolor=cyan,
}
\usepackage[capitalise]{cleveref}

\theoremstyle{definition}
\newtheorem{theorem}{Theorem} 
\newtheorem{proposition}[theorem]{Proposition} 

\newtheorem{lemma}[theorem]{Lemma}

\newtheorem{example}{Example}

\usepackage{graphicx}
\graphicspath{{images/}}

\usepackage{tensor}

\DeclarePairedDelimiterX{\inp}[2]{\langle}{\rangle}{#1, #2}

\usepackage{xparse}

\NewDocumentCommand\LH{mo}{%
  \IfNoValueTF{#2}
   {\mathcal{L}(\mathcal{H}^{#1})}
   {\mathcal{L}(\mathcal{H}^{#1},\mathcal{H}^{#2})}%
}

\newcommand\id{\leavevmode\hbox{\small1\kern-3.3pt\normalsize1}}

\newcommand{\sV}{\mathbb{V}}

\allowdisplaybreaks

\DeclareMathOperator\Log{Log}

\usepackage{bbold}

\usepackage{authblk}

\usepackage[title]{appendix}

\usepackage{mciteplus}

\title{Complex, Lorentzian, and Euclidean simplicial quantum gravity: numerical methods and physical prospects}

\author{Ding Jia (贾丁)\thanks{djia@perimeterinstitute.ca}}
\affil{Perimeter Institute for Theoretical Physics, Waterloo, Ontario, N2L 2Y5, Canada}
\affil{Department of Physics and Astronomy, University of Waterloo, Waterloo, Ontario, N2L 3G1, Canada}
\date{}

\begin{document}

\begin{CJK*}{UTF8}{gbsn}
\maketitle
\end{CJK*}

\begin{abstract}
Evaluating gravitational path integrals in the Lorentzian has been a long-standing challenge due to the numerical sign problem. We show that this challenge can be overcome in simplicial quantum gravity. By deforming the integration contour into the complex, the sign fluctuations can be suppressed, for instance using the holomorphic gradient flow algorithm. Working through simple models, we show that this algorithm enables efficient Monte Carlo simulations for Lorentzian simplicial quantum gravity. 

In order to allow complex deformations of the integration contour, we provide a manifestly holomorphic formula for Lorentzian simplicial gravity. This leads to a complex version of simplicial gravity that generalizes the Euclidean and Lorentzian cases. Outside the context of numerical computation, complex simplicial gravity is also relevant to studies of singularity resolving processes with complex semi-classical solutions. Along the way, we prove a complex version of the Gauss-Bonnet theorem, which may be of independent interest.
\end{abstract}

\section{Introduction}

To define a path integral, one needs to specify a way to enumerate the configurations to be summed over. For a non-relativistic particle, it is common is to introduce a lattice of discrete time steps, sum over piecewise linear paths across these steps, and take the continuum limit of lattice spaces going to zero \cite{Feynman1965QuantumIntegrals}. 

For gravity, one could similarly introduce a simplicial lattice, sum over piecewise flat geometries on the lattice characterized by the edge lengths, and take the limit of lattice refinement (\cref{fig:sll}). Historically, this method follows from Regge's insight \cite{Regge1961GeneralCoordinates} to use piecewise flat geometries to approximate curved space(times) at the classiccal level. Regge's classical approach is usually referred to as Regge calculus, or simplicial gravity, while the quantum path integral based on it is usually referred to as quantum Regge calculus, or simplicial quantum gravity \cite{Rocek1981QuantumCalculus, Williams1992ReggeBibliography, Loll1998DiscreteDimensions, Hamber2009QuantumApproach, Barrett2019TullioGravity}.

\begin{figure}
    \centering
    \includegraphics[width=1.\textwidth]{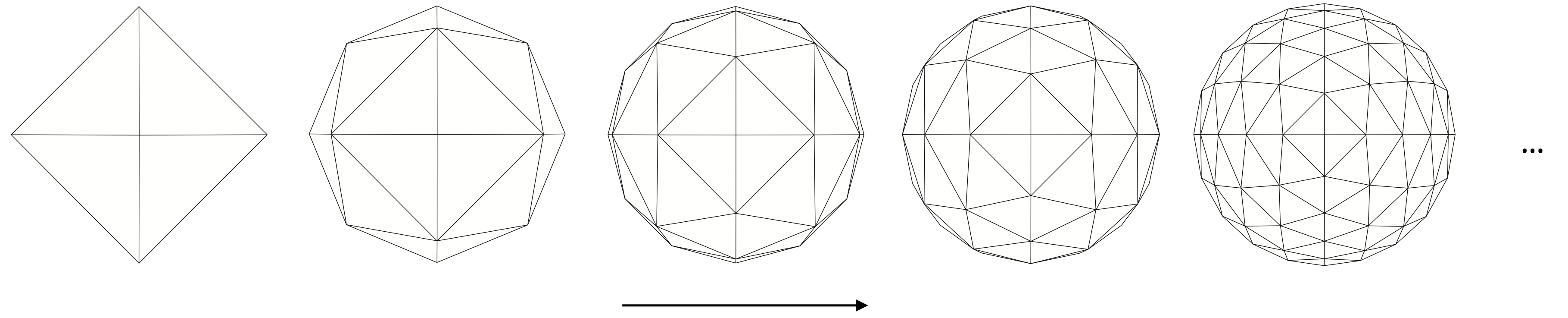}
    \caption{Simplicial lattice refinement.}
    \label{fig:sll}
\end{figure}

As a non-perturbative path integral approach, simplicial quantum gravity has a clear merit. It is known how to couple to the matter species of the Standard Model (see e.g., Chapter 6 of Hamber's textbook \cite{Hamber2009QuantumApproach} and references therein).

On the other hand, Euclidean quantum gravity faces the conformal instability problem \cite{Gibbons1977TheThermodynamics}. This is manifested as the problem of the spikes for Euclidean simplicial quantum gravity. In concrete $2D$ models, it is shown that configurations with diverging edge lengths dominate the path integral, even when the total spacetime area is bounded \cite{Ambjrn1997SpikesCalculusb}. One view is that only the weak coupling phase is rendered ill by the spiky configurations, but the strong coupling phase stays healthy \cite{Hamber2019VacuumGravity}. A more pessimistic view is that conformal instability poses a lethal threat to Euclidean simplicial quantum gravity. 

Whatever conformal instability actually implies about Euclidean quantum gravity, the case is different for the Lorentzian. For $2D$ simplicial quantum gravity it can be shown that the Lorentzian and Euclidean theories are inequivalent, and that spikes are absent in the Lorentzian where spacetime configurations are equipped with causal structures \cite{Tate2011Fixed-topologyDomain, Jia2022Time-spaceGravity}. \footnote{The proof of the absence of spikes in \cite{Tate2011Fixed-topologyDomain} assumes that the causal signature of simplicial lattice edges are fixed under the path integral. In \cite{Jia2022Time-spaceGravity} this assumption is dropped. It is shown that spikes are still absent, provided that causally irregular points with no lightcones attached are prohibited.} The question about higher dimensions is open, but the prospect that spikes are absent in the Lorentzian in general, and the fact that spacetime is Lorentzian in Nature form motivations to study Lorentzian simplicial quantum gravity. 

Apart from a few works \cite{Tate2011Fixed-topologyDomain, Tate2012Realizability1-simplex, MikovicPiecewiseGravity, Asante2021EffectiveGravity, Dittrich2022LorentzianSimplicial, Jia2022Time-spaceGravity}, the path integrals of Lorentzian simplicial quantum gravity have not been studied much in the past.\footnote{In this statement we mean by simplicial quantum gravity the formalism with dynamical lengths. The variation of simplicial quantum gravity with fixed lengths but dynamical lattice graphs has been extensively studied in the form of causal dynamical triangulation \cite{Ambjorn2012NonperturbativeGravity, Loll2020QuantumReview}.} 
Because of the numerical sign problem, naive Monte Carlo simulations do not work efficiently in the Lorentzian as in the Euclidean. This has remained a major obstacle for quantitative studies of Lorentzian simplicial quantum gravity.

In this work we propose to generalize simplicial quantum gravity to the complex domain. This allows us to apply the techniques of complex contour deformation developed in recent years to alleviate the sign problem \cite{AuroraScienceCollaboration2012HighThimble, AlexandruComplexProblem}. By a higher dimensional version of Cauchy's integration theorem, a path integral with a real integration contour can equally be evaluated along a complex contour if the two contours are related across a region where the integrand is holomorphic. The sign problem could be milder on the deformed contour. As reviewed in \cite{AlexandruComplexProblem}, this idea has been successfully applied to various lattice field theories of matter. It has also been applied to analyze gravitational propagators for spin-foam models in the large spin limit \cite{Han2021SpinfoamPropagator}.

Here we show that the complex contour deformation method also works for Lorentzian simplicial quantum gravity. Monte Carlo simulations are performed to compute the expectation value of spacetime lengths in $1+1D$ using the holomorphic gradient flow algorithm (also called the generalized thimble algorithm) \cite{Alexandru2016SignThimbles,  Alexandru2017MonteCarloModel, AlexandruComplexProblem}. It is found that the sign fluctuations are largely suppressed on suitable complex contours. As far as we know, this constitutes the first non-perturbative computation of Lorentzian simplicial gravitational path integrals. It opens the possibility to investigate questions about quantum gravity non-perturbatively and quantitatively using Lorentzian simplicial quantum gravity. 

Notably, the expectation values computed on the complex contours are directly the results of interest. There is no analytic continuation to Euclidean spacetime like in causal dynamical triangulation \cite{Ambjorn2012NonperturbativeGravity}, nor analytic continuation of parameters in the action like in causal sets \cite{Surya2019TheGravity}. These procedures face the open problem of inverse analytic continuation, which does not arise in the method used here.

Besides overcoming the sign problem, another reason to consider complex simplicial quantum gravity is to study singularity resolving processes. Quantum theory assigns non-zero probabilities to certain processes characterized by boundary conditions admitting not real, but complex semi-classical solutions. A standard example is particle tunneling \cite{Turok2014OnTime, ChermanReal-TimeInstantons, Tanizaki2014Real-timeTunneling}. It is conceivable that cosmological and black hole singularity resolving processes (see e.g., \cite{Frolov1981SphericallyGravity, Frolov1989ThroughUniverse, Barrabes1996HowHole, Frolov1998BlackPhysics, Vilenkin1982CreationNothing, Hartle1983WaveUniverse, Halliwell1991Introductory1990, Bojowald2001AbsenceCosmology, Modesto2004DisappearanceGravityb, Ashtekar2005BlackParadigmb, Hayward2006FormationHoles, Hossenfelder2010ConservativeProblem, Haggard2015Quantum-gravityTunneling, Barcelo2014TheResignation, Bianchi2018WhiteHoleb, DAmbrosio2021EndEvaporation, Oriti2017BouncingCondensatesb}) fall into the same category \cite{Hartle1989SimplicalModel, Li1993ComplexMachines, Gielen2015PerfectBounce, Gielen2016QuantumSingularities, Feldbrugge2017LorentzianCosmology, Dorronsoro2017RealCosmology,  Bramberger2017QuantumSingularities, Dittrich2022LorentzianSimplicial}. Lorentzian simplicial quantum gravity provides a formalism to compute the probabilities for such processes. To analyze the semi-classical solutions, the formalism needs to be generalized to the complex domain.

Although simplicial quantum gravity in the complex domain has been studied before  \cite{Hartle1989SimplicalModel, Louko1992ReggeCosmology, Birmingham1995LensCosmology, Birmingham1998ACalculus, Furihata1996No-boundaryUniverse, Silva1999SimplicialField, Silva1999AnisotropicField, Silva2000SimplicialPhi2, daSilvaWormholesMinisuperspace}, the complex theory is reached by analytically continuing the Euclidean theory. In addition, these works concentrated on symmetry-reduced models. 


In this work we specify Lorentzian simplicial gravity in arbitrary dimensions and without symmetry reduction with manifestly holomorphic expressions. Upon analytic continuation, the holomorphic expressions define simplicial gravity in the complex domain. The path integrals based on this complex action encompass both Lorentzian and Euclidean simplicial quantum gravity as special cases with different integration contours.


Along the way, we show that the celebrated Gauss-Bonnet theorem admits a complex generalization. This mathematical results may be of independent interest.

The paper is organized as follows. In \cref{sec:lav} and \cref{sec:a}, we review the geometric quantities of length, volume, and areas of Euclidean simplicial gravity, and generalize the quantities to the Lorentzian and complex domains. In \cref{sec:qg} we define simplicial gravitational path integrals in the Lorentzian and complex domains in terms of manifestly holomorphic expressions. In \cref{sec:hf} we review the holomorphic gradient flow algorithm for numerical computations of path integrals with complex actions. Starting in \cref{sec:2dsqg} we specialize to $2D$ simplicial quantum gravity and present the formulas needed for applying the holomorphic gradient flow algorithm. Along the way we prove a complex version of the Gauss-Bonnet theorem. In \cref{sec:nr} we present numerical results that overcome the sign problem. In \cref{sec:d} we finish with a discussion.

\section{Lengths and volumes}\label{sec:lav}

In simplicial gravity, the basic variable is the squared length, and the Einstein-Hilbert action is written in terms of volume and angles. (See Hamber's textbook \cite{Hamber2009QuantumApproach} for a comprehensive and lucid introduction to Euclidean simplicial quantum gravity.) In this section and next, we start by presenting length, volume and angles for simplicial geometry in the Euclidean domain, and then generalize these quantities to the Lorentzian and complex domains.

\subsection{Squared length as the basic variable}

Given a metric field $g_{ab}$ on a manifold, the \textbf{squared length} $\sigma$ of a line $\gamma$ segment is given by
\begin{align}\label{eq:slfg}
\sigma=\int_\gamma ds^2,
\end{align}
where $ds^2 = g_{ab} dx^a dx^b$ is the line element. 

In simplicial gravity each lattice edge $e$ has a squared lengths $\sigma_e$ with $\gamma$ taken along the edge. In the Euclidean domain, $\sigma\ge 0$. In the Lorentzian domain, we choose the signature convention that $\sigma>0$ for spacelike intervals, $\sigma<0$ for timelike intervals, and $\sigma=0$ for lightlike intervals.

In a continuum field theory, the basic gravitational variable is usually taken to be the metric field $g_{ab}$, and the squared length is derived from $g_{ab}$ using (\ref{eq:slfg}). In contrast, in simplicial gravity the basic variable is usually taken to be the squared lengths $\sigma$ on the lattice edges. A gravitational configuration is given in terms of the squared length on the edges, from which the metric can be derived as follows.

\begin{figure}
    \centering
    \includegraphics[width=.4\textwidth]{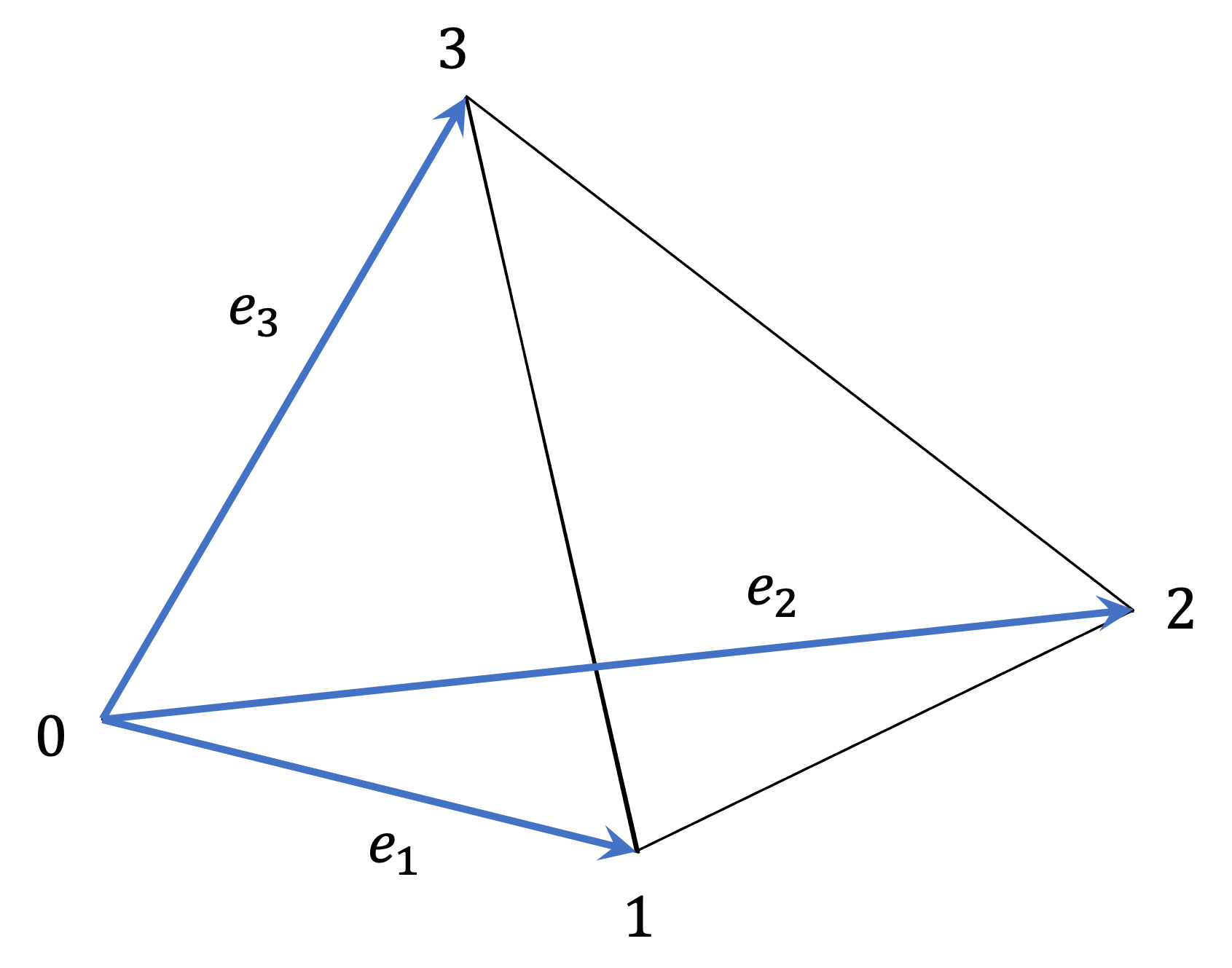}
    \caption{A simplex with labelled vertices $i$ and edge vectors $e_i$.}
    \label{fig:spl1}
\end{figure}

Let a $d$-simplex be given and label the vertices by $0,1,\cdots, d$ (\cref{fig:spl1}). Within the simplex we set up a coordinate system whose basis vectors $e_i$ for $i=1,\cdots, d$ point from vertex $0$ to vertex $i$. Define a dot product $\cdot$ by
\begin{align}\label{eq:edots}
e_i\cdot e_j = \frac{1}{2}(\sigma_{0i}+\sigma_{0j}-\sigma_{ij}),
\end{align}
where $\sigma_{ij}$ for $i,j = 0, 1,\cdots, d$ are the squared lengths of the edges connecting vertices $i$ and $j$. Using the metric
\begin{align}\label{eq:metric}
g_{ij} = \frac{1}{2}(\sigma_{0i}+\sigma_{0j}-\sigma_{ij}),
\end{align}
the dot product of any pair of vectors $u=  u^i e_i$ and $v=  v^i e_i$ can be computed as $u\cdot v =    g_{ij} u^i v^j$, where the Einstein summation convention is used. 


The metric (\ref{eq:metric}) is the simplicial analog of the continuum metric. In the continuum, squared lengths are computed through $ds^2=g_{ab}dx^a dx^b$. On a simplicial lattice, edge squared lengths are computed through
\begin{align}\label{eq:slv}
\sigma=v\cdot v =    g_{ij} v^i v^j,
\end{align}
where $v$ is the edge vector. For edges containing vertex $0$, $v=e_i$, and $v\cdot v=g_{ii}=\sigma_{0i}$. For other edges, $v=e_i-e_j$, and $v\cdot v=g_{ii}-g_{ij}-g_{ji}+g_{jj}=\sigma_{ij}$.

The simplex is understood to have a homogeneous interior. For a line segment within the simplex, the square length is computed by the same formula (\ref{eq:slv}) where $v$ is the vector for the line segment.

\subsection{Complexifying strategy}\label{sec:mtd}

In complexifying simplicial geometry, we adopt a ``squared length based'' methodology. After identifying a quantity of interest, such as volumes and angles, we express it as a function of the squared lengths. The function is chosen to agree with known expressions in the Lorentzian and/or Euclidean domains,
where the squared lengths take real values. In addition, the function should be holomorphic if possible to facilitate the deformations of integration contours when we study of the quantum theory. 


Suppose the above two requirements can be met. Then we can analytically continue the domain of the function to complex squared lengths.
When multi-valued functions such as the square root and the log are present, we will extend the domain to be the corresponding Riemann surfaces. 

As an example, consider the (linear) \textbf{length} defined by $l = \sqrt{\sigma}.$ This function is holomorphic away from the branch point $\sigma=0$. In the Euclidean domain $l>0$. In the Lorentzian domain $l>0$ for spacelike edges, and $l$ is positive imaginary for timelike edges in the current choice of the positive branch for the square root. 




\subsection{Volumes}\label{sec:vol}

The squared length and length given above are special cases of squared volumes and volumes. 

In the continuum, let $s$ be a simplex defined by some unit vectors. Suppose the metric is constant in the region of the simplex. Then the squared volume for the simplex is $\sV = \int_s \det g_{ab}(x) ~d^Dx=\frac{1}{d!}\det g_{ab}$, where $\frac{1}{d!}$ arises because this is for a simplex rather than a hypercube. 

On a simplicial complex, define the \textbf{squared volume} of a $d$-simplex by
\begin{align}\label{eq:svol1}
\sV = \frac{1}{(d!)^2}\det g_{ij},
\end{align}
where $g_{ij}$ as defined in (\ref{eq:metric}) is a function of the edge squared lengths.
An equivalent expression that is manifestly symmetric in the squared lengths is the Cayley-Menger determinant
\begin{align}\label{eq:svol}
\sV = \frac{(-1)^{d+1}}{2^d (d!)^2}
\begin{vmatrix}
 0 & 1 & 1 & 1 & \ldots  & 1 \\
 1 & 0 & \sigma _{01} & \sigma _{02} & \ldots  & \sigma _{0 d} \\
 1 & \sigma _{01} & 0 & \sigma _{12} & \ldots  & \sigma _{1 d} \\
 1 & \sigma _{02} & \sigma _{12} & 0 & \ldots  & \sigma _{2 d} \\
 \vdots  & \vdots  & \vdots  & \vdots  & \ddots & \vdots  \\
 1 & \sigma _{0 d} & \sigma _{1 d} & \sigma _{2 d} & \ldots  & 0 \\
\end{vmatrix}.
\end{align}
The \textbf{volume} $V$ of a $d$-simplex is defined by
\begin{align}\label{eq:vol}
V =\sqrt{\sV}.
\end{align}
Both $\sV$ and $V$ are defined for complex squared lengths. In (\ref{eq:vol}) the squared volume is taken to live on the Riemann surface of the square root function. $V$ is holomorphic as a function of the squared lengths away from the branch points where $\sV=0$. 

In the Euclidean domain, $\sV>0$. In the Lorentzian domain, $\sV \le 0$. The positive branch for the square root is chosen so that $V$ is positive imaginary or zero for Lorentzian simplices.

\begin{example}
In lower dimensions some familiar expressions are recovered. In $1D$ the volumes derived from (\ref{eq:svol}) and (\ref{eq:vol}) are
\begin{align}
\sV=&\sigma_{01},
\\
V =& \sqrt{\sigma_{01}},
\end{align}
which reproduce the length formulas.
In $2D$ the volumes for a triangle $t$ derived from (\ref{eq:svol}) and (\ref{eq:vol}) are
\begin{align}\label{eq:2dsvol}
\sV=&\frac{1}{16} \left(-\sigma _{01}^2-\sigma _{02}^2-\sigma _{12}^2+2  \sigma _{01} \sigma _{02}+2  \sigma _{01} \sigma _{12}+2 \sigma _{02} \sigma _{12}\right),
\\
V =& \frac{1}{4} \sqrt{-\sigma _{01}^2-\sigma _{02}^2-\sigma _{12}^2+2  \sigma _{01} \sigma _{02}+2  \sigma _{01} \sigma _{12}+2 \sigma _{02} \sigma _{12}},
\label{eq:2dvol}
\end{align}
which reproduce Heron's formula for triangle areas. 
\qed
\end{example}

\subsection{Generalized triangle inequalities}\label{sec:gti}

The squared distances must obey certain generalized triangle inequalities to describe Euclidean and Lorentzian simplices.

In the Euclidean domain, a simplex $s$ obeys
\begin{align}\label{eq:egti}
\sV > 0 \quad \text{ for all subsimplices of $s$ including $s$ itself}.
\end{align}
For example, for a triangle this means the squared area and the squared lengths are positive:
\begin{align}
&\sV=\frac{1}{16} \left(-\sigma _{01}^2-\sigma _{02}^2-\sigma _{12}^2+2  \sigma _{01} \sigma _{02}+2  \sigma _{01} \sigma _{12}+2 \sigma _{02} \sigma _{12}\right) >0,
\\ &\sigma_{01},\sigma_{02},\sigma_{03} >0.
\end{align}

\begin{figure}
    \centering
    \includegraphics[width=.4\textwidth]{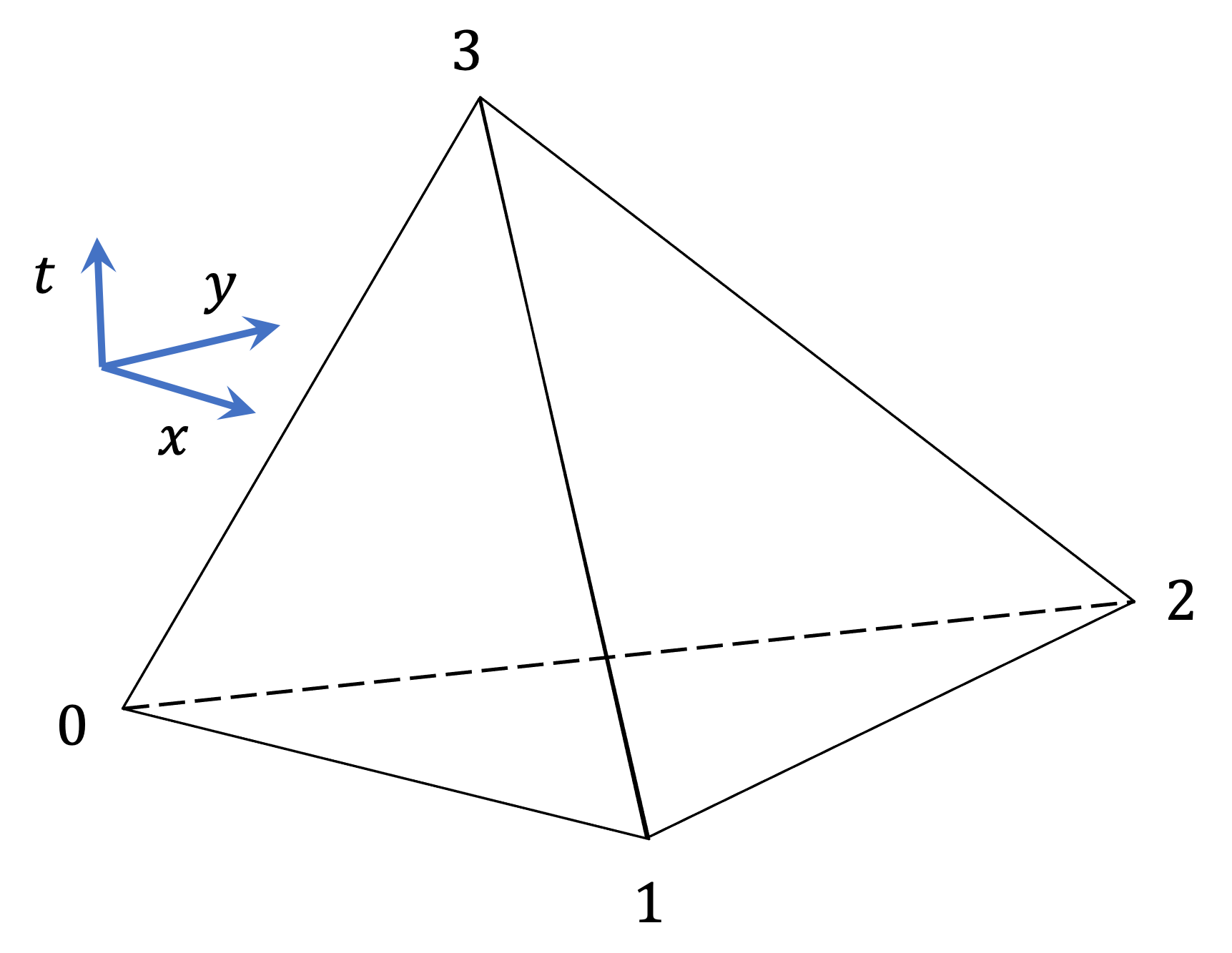}
    \caption{A $3D$ timelike simplex can have a spacelike subsimplex $012$ in addition to timelike subsimplices such as $013$.}
    \label{fig:spl3}
\end{figure} 

In the Lorentzian domain, a simplex $s$ obeys \cite{Tate2012Realizability1-simplex, Asante2021EffectiveGravity}
\begin{align}\label{eq:lgti}
\sV_s < 0; \text{ and } \sV_{r}<0 \implies \sV_{t}<=0 \text{ for all $t\supset r$}.
\end{align}
A simplex is timelike if $\sV<0$, and spacelike if $\sV>0$. In contrast to the Euclidean domain where all simplices and subsimplices have the same causal signature (spacelike), in the Lorentzian domain the subsimplices are allowed to be both timelike and spacelike \Cref{fig:spl3}. The Lorentzian generalized inequalities (\ref{eq:lgti}) first say that the simplex $s$ itself needs to be timelike. Furthermore, if any subsimplex $r$ is timelike, then all subsimplices $t$ containing $r$ cannot be spacelike. This is because a timelike subsimplex cannot be embedded in a spacelike subsimplex. For instance in \Cref{fig:spl3}, if the edge subsimplex $03$ is timelike, then the triangle subsimplices $013$ and $023$ containing the timelike edge $03$ must not be spacelike, which is a reasonable condition.

\section{Angles}\label{sec:a}

\subsection{Euclidean angles}\label{sec:ea}

In Euclidean space, what is the angle $\theta$ bounded by two vectors $a$ and $b$? Since
\begin{align}\label{eq:adotbth}
a\cdot b=|a||b|\cos\theta, \quad |x|:=\sqrt{x\cdot x},
\end{align}
one answer is that $\theta=\cos^{-1}\frac{a\cdot b}{|a||b|}$. 
Another answer is in terms of the scalar wedge product defined by
\begin{align}
a\wedge b=&\sqrt{(a\cdot b)^2-(a\cdot a)(b\cdot b)}.
\label{eq:awedgeb1}
\end{align}
Using $\sin^2 \theta+\cos^2 \theta=1$, it is easy to see that for $\theta>0$,
\begin{align}\label{eq:awedgebth}
a\wedge b=i |a||b|\sin \theta.
\end{align}
Therefore $\theta=\sin^{-1}\frac{a\wedge b}{i|a||b|}$.

The answer (\ref{eq:adotbth}) or (\ref{eq:awedgebth}) in isolation has ambiguities, because different angles can have the same $\cos$ or $\sin$ values. Within a $2\pi$ period, angles are uniquely determined when the information of $\cos^{-1}$ and $\sin^{-1}$ are combined. From (\ref{eq:adotbth}) and (\ref{eq:awedgebth}), we derive that $e^{i\theta}=\frac{1}{|a||b|}(a\cdot b+a\wedge b)$, so\footnote{This formula is related to the so-called ``geometric product'' $\vec{a}\cdot \vec{b}+\vec{a}\wedge \vec{b}$, which offers a way to encode rotations. The difference is that here $\vec{a}\wedge \vec{b}$ is a bivector instead of a scalar.}
\begin{align}\label{eq:ea}
\theta =& -i\log \alpha,
\\
\alpha=&\frac{1}{|a||b|}(a\cdot b+a\wedge b).
\end{align}
This determines $\theta$ uniquely within a $2\pi$ period depending on the choice of the branch for the log function.

\subsection{Complex angles}\label{sec:ca}

In the general complex domain, we take
\begin{align}
\theta =& -i\log \alpha,
\label{eq:theta}
\\
\alpha=&\frac{a\cdot b+a\wedge b}{\sqrt{a\cdot a}\sqrt{b\cdot b}}=\frac{a\cdot b+\sqrt{(a\cdot b)^2-(a\cdot a)(b\cdot b)}}{\sqrt{a\cdot a}\sqrt{b\cdot b}},
\label{eq:alpha}
\end{align}
as the definition of \textbf{complex angles}. Equation (\ref{eq:alpha}) is one of the expressions in Sorkin's definition of Lorentzian angles in the Minkowski plane \cite{SorkinLorentzianVectors}.\footnote{In Sorkin's definition of Lorentzian triangles \cite{SorkinLorentzianVectors}, (\ref{eq:alpha}) is used for angles bounded by two spacelike vectors in the same quadrant, and angles bounded by a spacelike vector and a timelike vector. A different expression is used for angles bounded by two timelike vectors in the same quadrant.}
Here we recognize that more generally, (\ref{eq:theta}) and (\ref{eq:alpha}) offer a unified definition for Euclidean, Lorentzian, and complex angles in all cases.\footnote{For the formula to apply to the Euclidean case, the $-i$ factor in (\ref{eq:theta}) is necessary. In comparing with other works based on Sorkin's definition one should keep in mind that the $-i$ factor is absent there. In addition, for Lorentzian angles (\ref{eq:la}) defined below differs in the choice of square root branches from Sorkin's formula.}

In terms of the edge squared lengths (\cref{fig:tri}),
\begin{align}
a\cdot b=& \frac{1}{2}(\sigma_{a}+\sigma_{b}-\sigma_c),
\label{eq:adotb}
\\
a\cdot a=&\sigma_a, \quad  b\cdot b=\sigma_b,
\\
a\wedge b=& \frac{1}{2} \sqrt{\sigma _{a}^2+\sigma _{b}^2+\sigma _{c}^2-2  \sigma _{a} \sigma _{b}-2  \sigma _{b} \sigma _{c}-2 \sigma _{c} \sigma _{a}}.
\label{eq:awedgeb}
\end{align}
Therefore
\begin{align}\label{eq:ca}
\theta =& -i\log \alpha,\nonumber
\\ \alpha=& \frac{\sigma_{a}+\sigma_{b}-\sigma_c+\sqrt{\sigma _{a}^2+\sigma _{b}^2+\sigma _{c}^2-2  \sigma _{a} \sigma _{b}-2  \sigma _{b} \sigma _{c}-2 \sigma _{c} \sigma _{a}}}{2\sqrt{\sigma_a}\sqrt{\sigma_b}}.
\end{align}
We take (\ref{eq:ca}) as the definition of \textbf{complex angles} in terms of complex squared lengths. This function is holomorphic away from the log and square root branch points. At the square root branch point of $a=0$ or $b=0$, the denominator becomes $0$. We will comment more on the (ir)relevance of this case in the end of \Cref{sec:la} and in \Cref{sec:cb}.

\begin{figure}
    \centering
    \includegraphics[width=.4\textwidth]{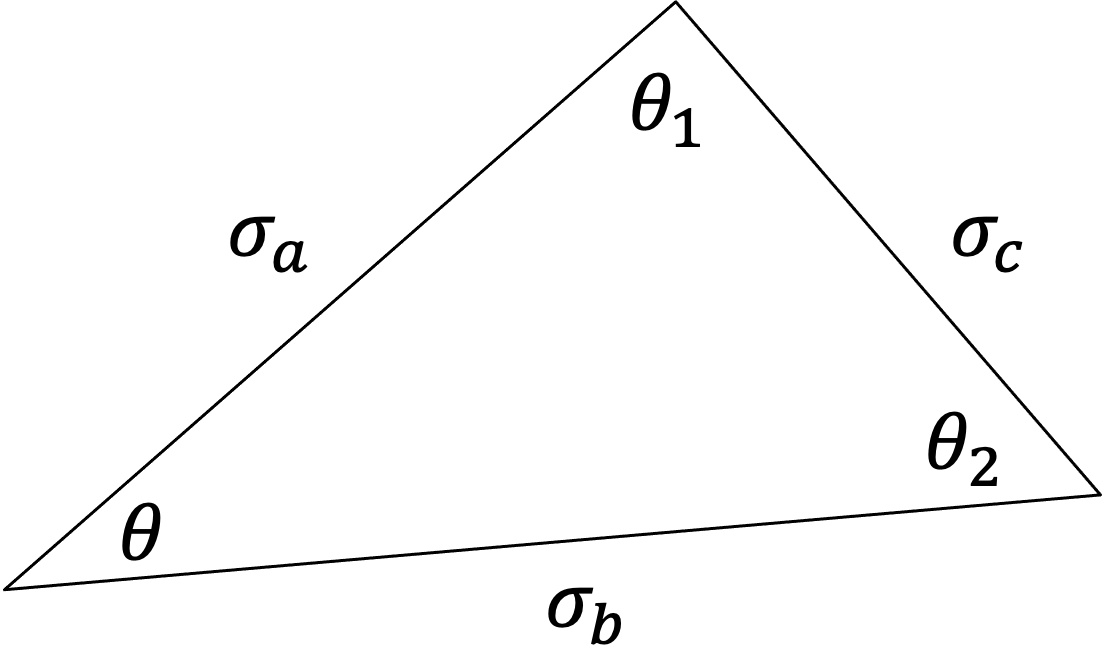}
    \caption{A triangle with squared lengths $\sigma_{a}, \sigma_{b}, \sigma_c$.}
    \label{fig:tri}
\end{figure} 

Note from (\ref{eq:2dsvol}) that the input $A$ to the numerator square root equals $-16\sV$, where $\sV$ is the squared volume for the triangle in \cref{fig:tri}. By the triangle inequalities of \cref{sec:gti}, $A>0$ for a Lorentzian triangle and $A<0$ for an Euclidean triangle.

For Euclidean angles the principal branches of the log and square root functions are chosen. The complex angles then reduce to the correct Euclidean angles, since the former are obtained by generalizing the latter. The choices of branches for Lorentzian angles are specified below.

\subsection*{Sum of complex angles in a triangle}

The angles of an Euclidean triangle sum to $\pi$. In the complex domain, this generalizes to $(2n+1)\pi$ with $n\in \mathbb{Z}$.
\begin{proposition}\label{th:sta}
The complex angles sum to $(2n+1)\pi$ with $n\in \mathbb{Z}$ for a triangle of complex squared edge lengths.
\end{proposition}
\begin{proof}
Consider a triangle with complex squared lengths $\sigma_{a}, \sigma_{b}, \sigma_c$ (\Cref{fig:tri}), and complex angles $\theta=-i\log\alpha, \theta_1=-i\log\alpha_1, \theta_2=-i\log\alpha_2$. A straightforward calculation using (\ref{eq:ca}) yields
\begin{align}
\alpha_1 \alpha_2=\frac{-\sigma_{a}-\sigma_{b}+\sigma_c+\sqrt{\sigma _{a}^2+\sigma _{b}^2+\sigma _{c}^2-2  \sigma _{a} \sigma _{b}-2  \sigma _{b} \sigma _{c}-2 \sigma _{c} \sigma _{a}}}{2\sqrt{\sigma_a}\sqrt{\sigma_b}}.
\end{align}
A similar calculation yields $\alpha \alpha_1 \alpha_2=-1$. For the complex log function, $\log(z_1 z_2)=\log z_1 +\log z_2$ up to multiples of $2\pi i$. Therefore $\theta+\theta_1+\theta_2=-i\log(-1)+2 \pi n=(2n+1)\pi$, where $n$ is an integer. 
\end{proof}

\subsection{Lorentzian angles}\label{sec:la}

In this section, we consider angles for Lorentzian simplicial geometries that obey the Lorentzian generalized triangle inequalities (\ref{eq:lgti}). In previous works \cite{SorkinLorentzianVectors, Tate2011Fixed-topologyDomain, Asante2021EffectiveGravity},
not one, but multiple expressions for Lorentzian angles in terms of log and trigonometric functions were used depending on where the edges lie in the Minkowski plane. A merit of the complex angle defined above is that it unifies these multiple cases (as well as the Euclidean case) in one formula. 

Here we focus on convex angles, because in simplicial gravity only these arise from individual simplices. Non-convex angles arise from summing the convex angles of individual simplices. Here we consider the branch choice 
\begin{align}
\theta =& -i \Log \alpha,\nonumber
\\
\alpha=&
\frac{a\cdot b+\sqrt{(a\cdot b)^2-(a\cdot a)(b\cdot b)}}{\sqrt{a\cdot a-0i}\sqrt{b\cdot b-0i}},\label{eq:la}
\end{align}
for Lorentzian angles. In terms of the squared lengths,
\begin{align}
\alpha=& \frac{\sigma_{a}+\sigma_{b}-\sigma_c+\sqrt{\sigma _{a}^2+\sigma _{b}^2+\sigma _{c}^2-2  \sigma _{a} \sigma _{b}-2  \sigma _{b} \sigma _{c}-2 \sigma _{c} \sigma _{a}}}{2\sqrt{\sigma_a-0i}\sqrt{\sigma_b-0i}}.
\end{align}
Here
\begin{align}\label{eq:lpb}
\Log z =& \log r + i\phi, \quad z=r e^{i\phi}\text{ with }\phi\in (-\pi,\pi]
\\\sqrt{z}=&\sqrt{r}e^{i\phi/2}, \quad z=r e^{i\phi}\text{ with }\phi\in (-\pi,\pi],
\label{eq:sqrtb}
\\\sqrt{z-0i}=&\sqrt{r}e^{i\phi/2}, \quad z=r e^{i\phi}\text{ with }\phi\in [-\pi,\pi).\label{eq:bcc}
\end{align}
The first two are just the principal branches of log and square root. The third one $\sqrt{z-0i}$ is negative imaginary for $z<0$. The symbol $-0i$ is a reminder that $z<0$ is continuously connected to $z>0$ through the lower complex plane instead of the upper one. 

The following properties hold for Lorentzian angles.
\begin{proposition}\label{prop:laa}
The Lorentzian convex angles $\theta$ defined by formula (\ref{eq:la}) are additive. 
\end{proposition}

\begin{proposition}\label{prop:caba}
The complex Lorentzian angle $\theta$ is related to the Lorentz boost angle $\theta_{\text{boost}}$ by
\begin{align}\label{eq:caba}
\theta = -i \theta_{\text{boost}}.
\end{align}
Here the convention is that $\theta_{\text{boost}}>0$ for a boost angle relating spacelike vectors, and $\theta_{\text{boost}}<0$ for a boost angle relating timelike vectors.
\end{proposition}

\begin{proposition}\label{prop:casl}
Between two edges related by the reflection across a light ray, the angle $\theta$ equals
\begin{align}
\theta = \pi/2,
\end{align}
whose imaginary part vanishes.
\end{proposition}

\begin{proposition}\label{prop:fp2p}
In the flat Minkowski plane, the angles around a point sum to $2\pi$.
\end{proposition}

\begin{proposition}\label{prop:rtheta}
For a convex Lorentzian angle $\theta$,
\begin{align}
\Re \theta = N\pi/2,
\end{align}
where $N=0,1,2$ is the number of light rays enclosed within the angle.
\end{proposition}
\begin{proposition}\label{prop:ltri}
The angles of a Lorentzian triangle sum to $2\pi$.
\end{proposition}

These results are easier to derive after working through some examples. These also serve to help readers unfamiliar with Lorentzian angles \cite{SorkinLorentzianVectors} to build some intuitions. 

In the Minkowski plane, a convex angle can bound $N=0,1,$ or $2$ light rays (\cref{fig:lpv}). According to whether the vectors bounding the angle are timelike or spacelike (for reasons mentioned below all the examples, we do not consider lightlike edges here), there are five cases in total. We consider them in turn.
\begin{example}[Spacelike edges within the same quadrant]
Consider spacelike edges $a$ and $b$ forming a triangle with squared lengths $\sigma_a=1, \sigma_b=3/4, \sigma_{ab}=-1/4$, where $\sigma_{ab}$ is the squared length for the third edge (\Cref{fig:lpv}). The complex angle $\theta$ bounded by $a$ and $b$ can be calculated using (\ref{eq:adotb}) to (\ref{eq:awedgeb}) as follows.
\begin{align}
a\cdot b=& \frac{1}{2}(\sigma_{a}+\sigma_{b}-\sigma_{ab})=1,
\\
a\cdot a=&\sigma_a=1, \quad  b\cdot b=\sigma_b=3/4,
\\
a\wedge b=&\sqrt{(a\cdot b)^2-(a\cdot a)(b\cdot b)}=1/2,
\\
\theta =& -i\log (\frac{a\cdot b+a\wedge b}{\sqrt{a\cdot a-0i}\sqrt{b\cdot b-0i}}) \nonumber
\\
=& -i\log\sqrt{3}.
\end{align}
\qed
\end{example}
The above calculation is based on the invariant quantity of the squared length and does not invoke any coordinate system. Alternatively, one could introduce a coordinate system in the Minkowski plane, represent $a$ and $b$ as vectors there, and use the Minkowski inner product for $\cdot$ to calculate $\theta$. For instance, one could choose $a=(1,0)$ and $b=(1,1/2)$ in the coordinate convention $(x,t)$. Then again $a\cdot b= 1^2-0=1$, $a\cdot a= 1^2-0^2=1$, and $b\cdot b=1^2-(1/2)^2=3/4$, so one will get the same result for $\theta$.

\begin{figure}
    \centering
    \includegraphics[width=.4\textwidth]{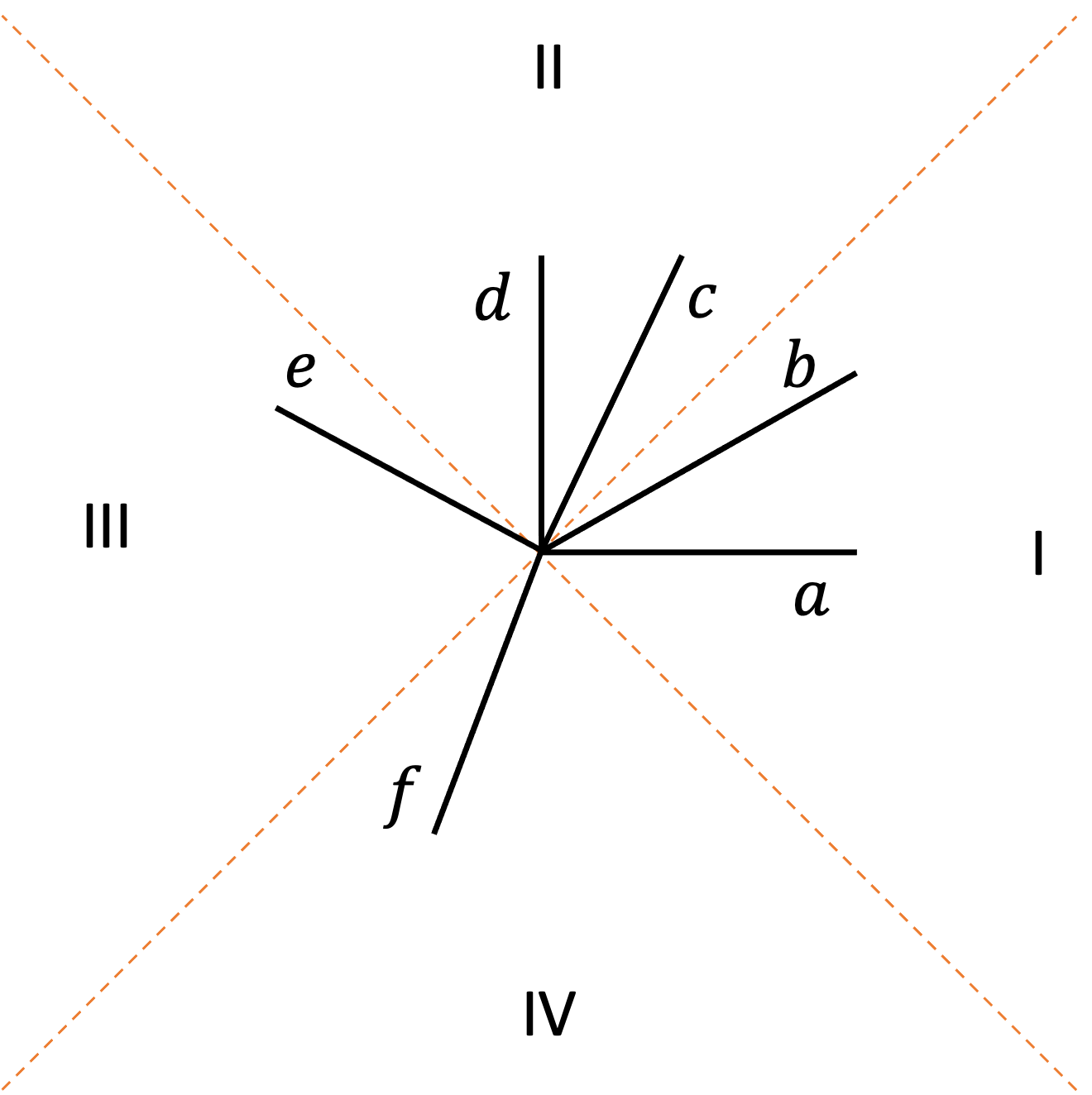}
    \caption{The Minkowski plane with four quadrants bounded by dashed light rays. The edges $a$ to $f$ are distributed in different quadrants.}
    \label{fig:lpv}
\end{figure}

The complex angle $\theta$ is related to the boost angle of Lorentz transformations. The boost angle from $a$ to $b$ is, up to a choice of sign,
\begin{align}\label{eq:ba}
\theta_{\text{boost}}=\cosh^{-1}(\hat{a}\cdot\hat{b}).
\end{align}
Here $\hat{x}:=x/\sqrt{\abs{x\cdot x}}$ denotes the normalized vector for $x$. For the vectors $a=(1,0)$ and $b=(1,1/2)$, $\hat{a}=(1,0)$ and $\hat{b}=(2/\sqrt{3},1/\sqrt{3})$. Therefore $\theta_{\text{boost}}=\cosh^{-1}(2/\sqrt{3})=\log\sqrt{3}$, where we used the elementary identity
\begin{align}\label{eq:chlog}
\cosh^{-1} z = \log(z+\sqrt{z^2-1}).
\end{align}
In this case for two spacelike edges, upon choosing the boost angle to be positive, we see that $\theta = -i \theta_{\text{boost}}.$ Using (\ref{eq:ba}) and (\ref{eq:chlog}), it is easy to check that this relation holds for all pairs of spacelike edges in the same quadrant in the Minkowski plane. We will see next that this relation also holds for timelike edges.

\begin{example}[Timelike edges within the same quadrant]
Consider timelike edges $c$ and $d$ forming a triangle with squared lengths $\sigma_c=-3/4, \sigma_d=-1, \sigma_{cd}=1/4$, where $\sigma_{cd}$ is the squared length for the third edge (\Cref{fig:lpv}). The complex angle $\theta$ bounded by $c$ and $d$ can be calculated using (\ref{eq:adotb}) to (\ref{eq:awedgeb}) as follows.
\begin{align}
c\cdot d=& \frac{1}{2}(\sigma_{c}+\sigma_{d}-\sigma_{cd})=-1,
\\
c\cdot c=&\sigma_c=-3/4, \quad  d\cdot d=\sigma_b=-1,
\\
c\wedge d=&\sqrt{(c\cdot d)^2-(c\cdot c)(d\cdot d)}=1/2,
\\
\theta =& -i\log (\frac{c\cdot d+c\wedge d}{\sqrt{c\cdot c-0i}\sqrt{d\cdot d-0i}}) \nonumber
\\
=& -i\log\frac{-1+1/2}{(-i\sqrt{3/4})(-i\sqrt{1})}=-i\log(1/\sqrt{3}).
\end{align}
\qed
\end{example}
Alternatively, setting $c=(1/2,1)$ and $d=(0,1)$ in a coordinate system $(x,t)$ and performing the calculation there leads to the same $\theta$.

Note that $c$ and $b$, as well as $d$ and $a$ are related by reflection with respect to the light ray separating quadrant I and II. The same Lorentz boost transformation that maps $a$ to $b$ will map $d$ to $c$. The boost angle from $a$ to $b$ is anti-clockwise, while that from $d$ to $c$ is clockwise. Since we chose the boost angle from $a$ to $b$ to be positive, it is reasonable to choose the boost angle from $d$ to $c$ to be negative. In this case we have
\begin{align}\label{eq:ba2}
\theta_{\text{boost}}=-\cosh^{-1}(|\hat{c}\cdot\hat{d}|)=-\cosh^{-1}(-\hat{c}\cdot\hat{d}),
\end{align}
since for timelike vectors in the same quadrant $\hat{c}\cdot\hat{d}<0$, and the normalized vectors take the form $\hat{x}:=x/\sqrt{\abs{x\cdot x}}=x/\sqrt{-x\cdot x}$.
From this we obtain $\hat{c}=(1/\sqrt{3},2/\sqrt{3})$ and $\hat{d}=(0,1)$, so $\theta_{\text{boost}}=-\cosh^{-1}(2/\sqrt{3})=\log(1/\sqrt{3})$. 

Again, $\theta = -i \theta_{\text{boost}}$. Using (\ref{eq:ba2}) and (\ref{eq:chlog}), it is not hard to check that actually this relation holds for all pairs of timelike edges in the same quadrant in the Minkowski plane. Since boost angles exist only between two spacelike vectors in the same quadrant and two timelike vectors in the same quadrant, we have 
proved \Cref{prop:caba}.

\begin{example}[A spacelike edge and a time like edge]
Consider the spacelike edge $a$ and timelike edge $c$ forming a triangle with squared lengths $\sigma_a=1, \sigma_c=-3/4, \sigma_{ac}=-3/4$, where $\sigma_{ac}$ is the squared length for the third edge (\Cref{fig:lpv}). The complex angle $\theta$ bounded by $a$ and $c$ can be calculated using (\ref{eq:adotb}) to (\ref{eq:awedgeb}) as follows.
\begin{align}
a\cdot c=& \frac{1}{2}(\sigma_{a}+\sigma_{c}-\sigma_{ac})=1/2,
\\
a\cdot a=&\sigma_a=1, \quad  c\cdot c=\sigma_c=-3/4,
\\
a\wedge c=&\sqrt{(a\cdot c)^2-(a\cdot a)(c\cdot c)}=1,
\\
\theta =& -i\log (\frac{a\cdot c+a\wedge c}{\sqrt{a\cdot a-0i}\sqrt{c\cdot c-0i}}) \nonumber
\\
=& -i\log\frac{1/2+1}{(\sqrt{1})(-i\sqrt{3/4})}=-i\log(i\sqrt{3})=-i\log\sqrt{3}+\pi/2.
\end{align}
\qed
\end{example}
Alternatively, setting $a=(1,0)$ and $c=(1/2,1)$ in a coordinate system $(x,t)$ and performing the calculation there leads to the same $\theta$.

Note the relevance of the choice of branch for the square root. Had we chosen the branch without $-0i$, the denominator would be $i\sqrt{3/4}$ instead, and the real part of $\theta$ would be $- \pi/2$. In the choice with $-0i$, we have:
\begin{lemma}\label{lm:aste}
The angle $\theta$ between a spacelike edge and a timelike edge obeys
\begin{align}
\Re\theta = \pi/2.
\end{align}
\end{lemma}
\begin{proof}
Without loss of generality let $\sigma_a>0$ and $\sigma_c<0$. Then $(a\cdot a)(c\cdot c)<0$, so $a\wedge c=\sqrt{(a\cdot c)^2-(a\cdot a)(c\cdot c)}>\abs{a\cdot c}$. Therefore the numerator of $\alpha$, $a\cdot c+a\wedge c$, is positive. The denominator $\sqrt{a\cdot a-0i}\sqrt{c\cdot c-0i}$ is negative imaginary. Therefore $\alpha$ is positive imaginary. It follows that
$\theta=-i\log(i r)=-i\log r+\pi/2$ for some $r>0$.
\end{proof}

For the special case of two edges $a$ and $c$ related by a reflection across a light ray as the reflection axis, the angle bounded by them equals $\theta = \pi/2$. This is the content of \Cref{prop:casl}, which is proved by noting that $\sigma_a=-\sigma_c$ and $\sigma_{ac}=0$. From these we derive that $a\cdot c = 0$, $a\wedge c= \sqrt{\sigma_a^2}$, whence $\alpha=\sqrt{\sigma_a^2}/(\sqrt{\sigma_a-0i}\sqrt{-\sigma_a-0i})=i$. Therefore $\theta=-i\log\alpha=\pi/2$. 

This should be expected. The boost angles from $a$ and $c$ to the light ray are equal in magnitude and opposite in sign. When added up to obtain $\Im\theta$ according to \Cref{prop:caba}, they cancel. By \Cref{prop:rtheta}, $\Re\theta=\pi/2$ because travelling from $a$ to $c$ crosses one light ray. 

\Cref{prop:casl} implies that the in the flat Minkowski plane the angles around around a point sum to $2\pi$, which is the content of \Cref{prop:fp2p}. Consider four edges right in the middle of the four quadrants. According to \cref{prop:casl}, the four angles formed by them all equal $\pi/2$, so they sum to $2\pi$.

\begin{example}[Spacelike edges in different quadrants]
Consider two spacelike edges $a$ and $e$ in different quadrants forming a triangle with squared lengths $\sigma_a=1, \sigma_e=3/4, \sigma_{ae}=15/4$, where $\sigma_{ae}$ is the squared length for the third edge (\Cref{fig:lpv}). The complex angle $\theta$ bounded by $a$ and $e$ can be calculated using (\ref{eq:adotb}) to (\ref{eq:awedgeb}) as follows.
\begin{align}
a\cdot e=& \frac{1}{2}(\sigma_{a}+\sigma_{e}-\sigma_{ae})=-1,
\\
a\cdot a=&\sigma_a=1, \quad  e\cdot e=\sigma_e=3/4,
\\
a\wedge e=&\sqrt{(a\cdot e)^2-(a\cdot a)(e\cdot e)}=1/2,
\\
\theta =& -i\log (\frac{a\cdot e+a\wedge e}{\sqrt{a\cdot a-0i}\sqrt{e\cdot e-0i}}) \nonumber
\\
=& -i\log(-1/\sqrt{3})=-i\log(1/\sqrt{3})+\pi.
\end{align}
\qed
\end{example}
Again, the readers can check that the vectors $a=(1,0)$ and $e=(-1,1/2)$ leads to the same $\theta$. 

Note the relevance of the choice of branch for the log function. The principal branch which we chose yields $\Re\theta = \pi$ for $\alpha<0$. A different choice could result in $\Re\theta = -\pi$. Given the branch choices for the square roots, only for the principal branch can the angles possibly be additive. To see this, note that by \Cref{lm:aste}, each light ray crossing accrues $\pi/2$ for $\Re\theta$. Since from $a$ to $e$ there are two light rays crossed, $\Re\theta$ needs to be $\pi$ if the angles are additive. 

\begin{example}[Timelike edges in different quadrants]
Consider two timelike edges $d$ and $f$ in different quadrants forming a triangle with squared lengths $\sigma_d=-1, \sigma_f=-3/4, \sigma_{df}=-15/4$, where $\sigma_{df}$ is the squared length for the third edge (\Cref{fig:lpv}). The complex angle $\theta$ bounded by $d$ and $f$ can be calculated using (\ref{eq:adotb}) to (\ref{eq:awedgeb}) as follows.
\begin{align}
d\cdot f=& \frac{1}{2}(\sigma_{d}+\sigma_{f}-\sigma_{df})=1,
\\
d\cdot d=&\sigma_d=-1, \quad  f\cdot f=\sigma_f=-3/4,
\\
d\wedge f=&\sqrt{(d\cdot f)^2-(d\cdot d)(f\cdot f)}=1/2,
\\
\theta =& -i\log (\frac{d\cdot f+d\wedge f}{\sqrt{d\cdot d-0i}\sqrt{f\cdot f-0i}}) \nonumber
\\
=& -i\log\frac{1+1/2}{(-i)(-i\sqrt{3/4})}=-i\log(-\sqrt{3})=-i\log(\sqrt{3})+\pi.
\end{align}
\qed
\end{example}
Alternatively, setting $d=(0,1)$ and $f=(-1/2,-1)$ in a coordinate system $(x,t)$ and performing the calculation there leads to the same $\theta$.

In the above two cases $\Re \theta = \pi$. It actually holds in general that crossing two light rays makes the angle accrue a real part of $\pi$. The reason is that the log argument is negative for two light ray crossings, which yields $\Re \theta = \pi$. To see that the log argument is negative, note that for two spacelike vectors $a$ and $e$ in different quadrants, $a\cdot e= \frac{1}{2}(\sigma_{a}+\sigma_{e}-\sigma_{ae})<0$ as a consequence of the Lorentzian triangle inequality (\ref{eq:lgti}). Therefore the log argument $\frac{a\cdot e+a\wedge e}{\sqrt{a\cdot a-0i}\sqrt{e\cdot e-0i}}<0$. For two timelike vectors $d$ and $f$ in different quadrants, $d\cdot f= \frac{1}{2}(\sigma_{d}+\sigma_{f}-\sigma_{df})>0$ as a consequence of the Lorentzian triangle inequality (\ref{eq:lgti}). In addition, $d\wedge f=\sqrt{(d\cdot f)^2(d\cdot d)(f\cdot f)}<\abs{d\cdot f}$. Therefore the log argument $\frac{d\cdot f+d\wedge f}{\sqrt{d\cdot d-0i}\sqrt{f\cdot f-0i}}<0$.

Since a convex angle in the Minkowski plane can only enclose $0,1$ or $2$ light rays, we have proved \Cref{prop:rtheta}. For any triangle in the Minkowski plane, the three angles enclose two light rays in total. Therefore the sum of the three angles have $\pi$ as the real part. By \cref{th:sta}, the imaginary part vanishes. This proves \cref{prop:ltri}.

Finally, we want to prove \cref{prop:laa}, i.e., 
\begin{align}
\theta(a,c)=\theta(a,b)+\theta(b,c),
\end{align}
where $b$ lies between $a$ and $c$ in the Minkowski plane, and $\theta(x,y)=-i\log \alpha(x,y)$ denotes the convex angle defined by some vectors $x$ and $y$ according to (\ref{eq:la}). The first part of the proof is the same as Sorkin's proof for his equation (3) in \cite{SorkinLorentzianVectors}. Explicitly, since the angles are convex and $b$ lies in between $a$ and $c$, one could write $b=\alpha a +\beta c$ with $\alpha, \beta\ge 0$. This can be plugged in
\begin{align}
(\frac{a\cdot b+a\wedge b}{\sqrt{a\cdot a}\sqrt{b\cdot b}})( \frac{b\cdot c+b\wedge c}{\sqrt{b\cdot b}\sqrt{c\cdot c}})= \frac{a\cdot c+a\wedge c}{\sqrt{a\cdot a}\sqrt{c\cdot c}},
\end{align}
i.e., $\alpha(a,b)\alpha(b,c)=\alpha(a,c)$, to eliminate $b$ and establish the identity.

For the complex log function, $\theta(a,b)+\theta(b,c)=-i\log \alpha(a,b)-i\log \alpha(b,c)=-i \log (\alpha(a,b)\alpha(b,c))= -i\log \alpha(a,c)=\theta(a,c)$ up to an integer multiple of $2\pi$. However, by \Cref{prop:rtheta} and the assumption that all three angles are convex, the real part of the left hand side can only be $0,\pi/2,$ or $\pi$. The same holds for the right hand side. Therefore the multiple of $2\pi$ has to be zero, and we established $\theta(a,b)+\theta(b,c)=\theta(a,c)$.

\subsection*{Lightlike edges}

When one or two of the edges that bound the angle are lightlike, the Lorentzian angle defined in (\ref{eq:la}) could diverge. In \cite{SorkinLorentzianVectors}, special care is taken to redefine such angles. 

We will not perform any redefinition for angles with lightlike edges in this work, because the main focus is on the quantum theory. In the path integral, squared lengths is integrated over for each edge. Zero (lightlike) squared length is of measure zero, and a special redefinition just on this measure zero set is not necessary. 
See \cref{sec:cb} for additional discussions on the (ir)relevance of lightlike edges for the gravitational path integral.


\subsection{Dihedral angles}\label{sec:da}

\begin{figure}
    \centering
    \includegraphics[width=.5\textwidth]{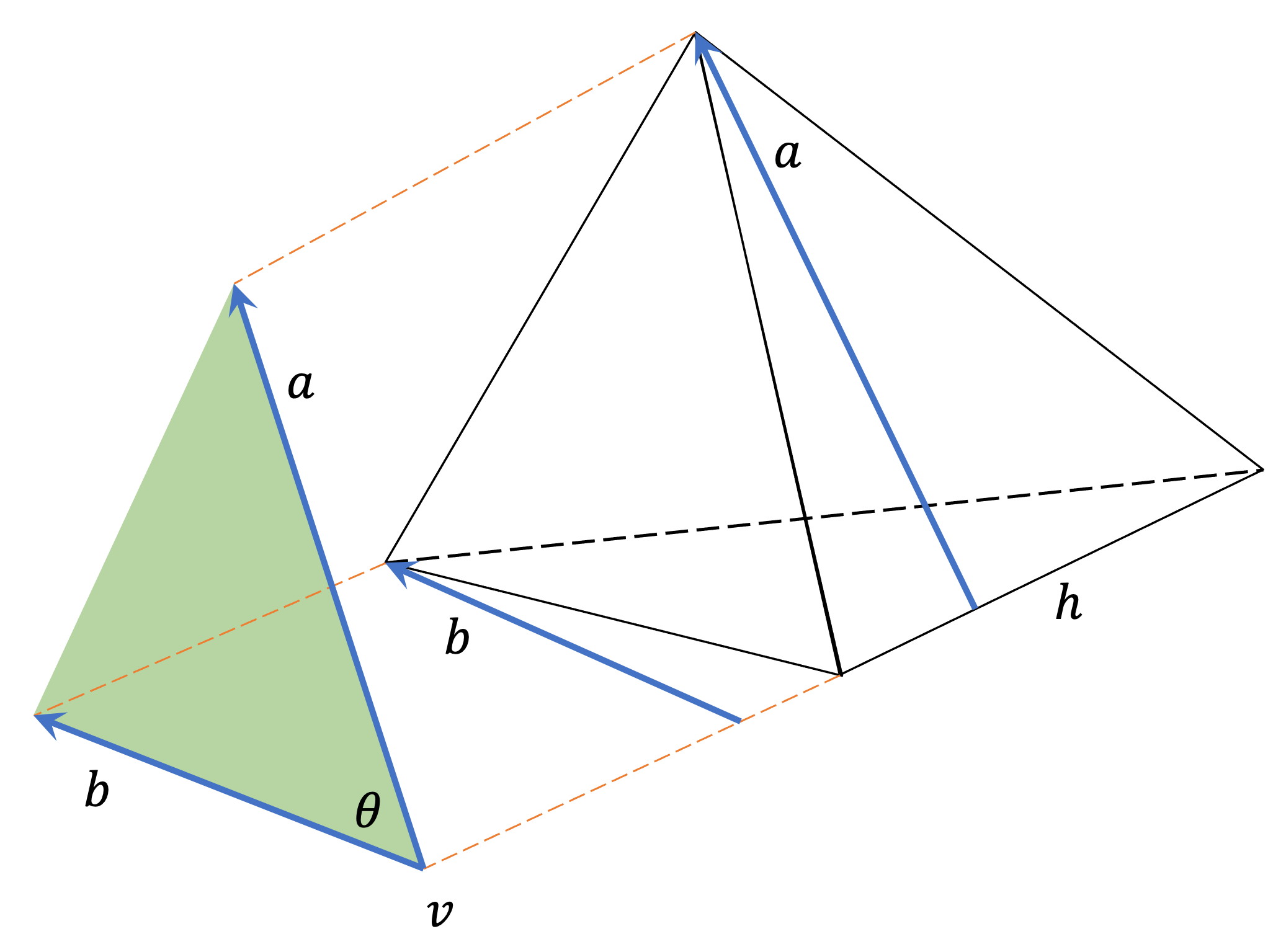}
    \caption{In $3D$, the tetrahedron simplex $s$ projects into the shaded triangle orthogonal to the hinge edge $h$. The dihedral angle $\theta_{s,h}$ projects to the triangle angle $\theta$. The faces bounding the dihedral angle project to the edges $a$ and $b$ of the triangle.}
    \label{fig:da}
\end{figure} 

In simplicial gravity, curvature is captured by deficit angles, which is in turn defined in terms of dihedral angles. 

A dihedral angles is formed by two codimension-$1$ faces at a hinge, which is a codimension-$2$ simplex. For instance in $2D$, the dihedral angle $\theta_{s,h}$ in triangle $s$ at vertex $h$ is the angle formed by the two edges sharing $h$. In $3D$ the dihedral angle $\theta_{s,h}$ in tetrahedron $s$ at edge $h$ is the angle formed by the two triangles sharing $h$. In $4D$ the dihedral angle $\theta_{s,h}$ in $4$-simplex $s$ at triangle $h$ is the angle formed by the two tetrahedrons sharing $h$ etc.

As illustrated in \Cref{fig:da}, dihedral angles can be obtained by projecting $s$ to the triangle orthogonal to $h$, and extracting the triangle angle at the vertex that $h$ projects to. Using (\ref{eq:alpha}), namely
\begin{align}
\theta =& -i\log \alpha,\quad
\alpha=\frac{a\cdot b+\sqrt{(a\cdot b)^2-(a\cdot a)(b\cdot b)}}{\sqrt{a\cdot a}\sqrt{b\cdot b}},
\end{align}
the dihedral angle can be computed from $a\cdot b$, $a\cdot a$, and $b\cdot b$ of the projected triangle. However, in simplicial gravity the input data are the squared distances $\sigma_e$ on the simplicial edges $e$. We need to express $a\cdot b$, $a\cdot a$, and $b\cdot b$ in terms $\sigma_e$.


\subsection*{Volume forms}

To express $a\cdot b$, $a\cdot a$, and $b\cdot b$ in terms $\sigma_e$, it is useful to introduce a volume form representation of the (sub)simplices \cite{Hartle1985SimplicialDiscussion}. An $n$-simplex has $n+1$ vertices. With one of the vertices labelled as $0$, the $n$ vectors $e_i, i=1,\dots, n$ starting from $0$ and pointing to the other $n$ vertices characterize the simplex (\cref{fig:spl1}). 

In \cref{sec:vol} we treated $e_i$ as the basis vectors in defining the metric $g_{ij}$ which equals $e_i\cdot e_j$. Let $e^i$ be the dual vectors so that $e^i(e_j)=\delta^i_j$. A $d$-simplex $s$ can represented by the $d$-form
\begin{align}
\omega_s = e^1\wedge \cdots \wedge e^d.
\end{align}
Then an $n$-dimensional subsimplex $r$ with edge vectors $e_{r_1},\dots, e_{r_n}$ can be represented by the $n$-form
\begin{align}
\omega_r = e^{r_1}\wedge \cdots  \wedge e^{r_n}.
\end{align}
The ordering of the indices $r_i$ decides an orientations for the (sub)simplex.

The dot product of two $n$-forms is given by
\begin{align}\label{eq:fdot}
\omega_r\cdot \omega_t = (\frac{1}{n!})^2 \det(e_{r_i}\cdot e_{t_j}). 
\end{align}
\cref{eq:fdot} conforms to the standard definition of inner products for $n$-forms. One can check that if $e_{r_i}=e_{r_j}$ for any $i\ne j$, or if $e_{t_i}=e_{t_j}$ for any $i\ne j$, then $\omega_r\cdot \omega_t=0$, which should hold for forms. 
By the definition (\ref{eq:svol1}) of the squared volume,
\begin{align}\label{eq:fv2}
\omega_r\cdot \omega_r = (\frac{1}{n!})^2 \det(e_{r_i}\cdot e_{r_j}) = \sV_r.
\end{align}

\subsection*{Vector dot products}

The form representation can be used to express $a\cdot b$, $a\cdot a$, and $b\cdot b$ for the dihedral angle in terms $\sigma_e$. Let $\omega_h$ be the $d-2$-form of the hinge $h$, and let
\begin{align}
\omega_{a} = \omega_h\wedge e, \quad \omega_{b} = \omega_h\wedge e'
\end{align}
be the $d-1$-forms of the faces of $a$ and $b$ (one might change the order between $\omega_h$ and $e$ ($e'$) if a different orientation is suitable). The edge vector $e$ can be written as $e=a+e_{\parallel}$, where $a$ is orthogonal to $h$ and $e_{\parallel}$ is parallel to $h$. Similarly $e'=b+e_{\parallel}'$. Since $e_{\parallel}$ and $e_{\parallel}'$ are parallel to $h$, it follows from the properties of forms that $\omega_{a} = \omega_h\wedge a$ and $\omega_{b} = \omega_h\wedge b$. Therefore
\begin{align}
\omega_{a} \cdot \omega_{b} =& (\omega_h\wedge a) \cdot (\omega_h\wedge b)
\\=& \frac{\omega_{h} \cdot \omega_{h}}{(d-1)^2} ~ a\cdot b.
\end{align}
In the second line we used the definition (\ref{eq:fdot}) and noted that since $a$ and $b$ are orthogonal to $h$, $a\cdot e=b\cdot e=0$ for any $e$ of $h$. 

Therefore
\begin{align}\label{eq:adbf}
a\cdot b = (d-1)^2 ~\frac{\omega_{a} \cdot \omega_{b}}{\omega_{h} \cdot \omega_{h}}.
\end{align}
The other terms $a\cdot a$ and $b\cdot b$ can be obtained by setting $a=b$. The numerator of (\ref{eq:adbf}) can be expressed in squared lengths using
\begin{align}
\omega_{a} \cdot \omega_{b} =& \frac{1}{(d-1)!^2} \det(e_{a_i}\cdot e_{b_j})
\label{eq:ipab1}
\\=& \frac{1}{(d-1)!^2} \det(\frac{1}{2}(\sigma_{0 a_i}+\sigma_{0 b_j}-\sigma_{a_i b_j})),
\label{eq:ipab2}
\end{align}
where (\ref{eq:fdot}) and (\ref{eq:edots}) are used. Here $a_i$ is the $i$-th vertex of the subsimplex $a$, and $b_j$ is the $j$-th vertex of the subsimplex $b$. The vertex $0$ is the one fixed when specifying the $d$-simplex $s$, and the squared lengths $\sigma_{0 a_i}, \sigma_{0 b_j}, \sigma_{a_i b_j}$ are inputs to simplicial gravity. According to (\ref{eq:fv2}), the denominator $\omega_{h} \cdot \omega_{h}$ of (\ref{eq:adbf}) simply equals $\sV_h$, which is a function of squared lengths by definition (\ref{eq:svol1}) or (\ref{eq:svol}). These formulas can then be used to express the dihedral angles in terms of squared lengths.

Incidentally, there is an alternative useful expression
\begin{align}
\omega_{a} \cdot \omega_{b} = d^2 \pdv{\sV}{\sigma_e},
\end{align}
where $e$ is the edge whose vertices are outside the hinge $h$ common to subsimplices $a$ and $b$. This expression can be derived using (\ref{eq:svol1}), (\ref{eq:ipab1}), (\ref{eq:edots}) and (\ref{eq:metric}).

\subsection{Deficit angles}

In simplicial gravity, curvature is captured by deficit angles. The deficit angle at a hinge is the difference between the flat space(time) value and the actual value for the sum of dihedral angles around the hinge.

At a hinge $h$ in the interior of a region (instead of on the boundary), the deficit angle is defined as 
\begin{align}\label{eq:da1}
\delta_h =& 2\pi - \sum_{s\ni h}\theta_{s, h},
\end{align}
where the sum is over all simplices $s$ containing $h$. 

Here $2\pi$ is the flat space(time) value. The dihedral angles around $h$ can be obtained by projecting the simplices to the plane orthogonal to $h$ and summing the angles around the point $h$ projects to (\cref{sec:da}). In flat Euclidean space, the angles obviously sum to $2\pi$. In flat Lorentzian spacetime, they also sum to $2\pi$ according to \cref{prop:fp2p}. In the complex domain it is taken as an assumption that the flat value is $2\pi$, so that (\ref{eq:da1}) constitutes a definition of the complex deficit angle in general.

If the hinge $h$ lies on the boundary of a region, the dihedral angles around it within that region can sum to less than $2\pi$ for the flat case. Suppose there are $Q_h$ regions sharing the hinge $h$. Then one way to define the deficit angle is
\begin{align}\label{eq:da2}
\delta_h =& \frac{2\pi}{Q_h} - \sum_{s\ni h}\theta_{s, h}.
\end{align}
This ensures additivity, i.e., once all the deficit angles in all regions are summed over (\ref{eq:da1}) is recovered. \Cref{eq:da2} is taken as the general definition of the \textbf{complex deficit angle}, with $Q_h=1$ if $h$ lies in the interior of the region.


\section{Quantum gravity}\label{sec:qg}

Formally, gravitational path integrals take the form
\begin{align}\label{eq:qgl}
Z=\int \mathcal{D}g ~ e^{i \int d^dx \sqrt{-g} ( - \lambda +k R + \cdots)}
\end{align}
in the Lorentzian, and 
\begin{align}\label{eq:qge}
Z=\int \mathcal{D}g ~ e^{\int d^dx \sqrt{g} ( - \lambda + k R+ \cdots)}
\end{align}
in the Euclidean. The dots stand for higher order terms that may be present. Here the Riemann tensor convention is
\begin{align}\label{eq:rt}
{R^{\rho }}_{\sigma \mu \nu }=\partial _{\mu }\Gamma _{\nu \sigma }^{\rho }-\partial _{\nu }\Gamma _{\mu \sigma }^{\rho }+\Gamma _{\mu \lambda }^{\rho }\Gamma _{\nu \sigma }^{\lambda }-\Gamma _{\nu \lambda }^{\rho }\Gamma _{\mu \sigma }^{\lambda },
\end{align}
so that as usual $\lambda>0$ leads to a De Sitter spacetime in cosmology.

To give an exact meaning to these formal expressions non-perturbatively, one needs to specify a way to enumerate gravitational configurations to be summed over.


\subsection{Simplicial quantum gravity}

In simplicial quantum gravity, 
\begin{align}\label{eq:sqg1}
Z=\int_C \mathcal{D}\sigma ~ e^{E[\sigma]}, 
\end{align}
where the exponent $E$ is given below. The gravitational configurations are specified by the squared lengths $\sigma$ on edges of simplicial lattices, and the path integral measure takes the form
\begin{align}\label{eq:sqgm1}
\int_C \mathcal{D}\sigma = (\sum_\tau) \lim_\Gamma \prod_{e\in\Gamma} \int_{-\infty}^\infty d\sigma_e ~ \mu[\sigma]  C[\sigma].
\end{align}
The meaning of the new symbols are explained in the next several paragraphs.

The integration measure factor $\mu[\sigma]$ is not known \textit{a priori}. Suppose one wants to define the path integral so that even on a finite lattice (without taking the lattice refinement limit) the result is exact result. Then one idea for fixing the measure is to demand discretization independence \cite{Dittrich2011PathGravityb}. This would lead to a non-local measure in $4D$ \cite{Dittrich2014DiscretizationGravity}. Alternatively, one could adopt simpler local measures and demand that the exact result be obtained only after taking the lattice refinement limit. In this case different measures could belong to a same universality class and lead to the same result in the lattice refinement limit \cite{Hamber2009QuantumApproach}. However, there seems to be no consensus exactly which measures are correct to be used. In analogy to the continuum measures factors $(\det g)^m$, a commonly used family of simplicial measures is the product of powers of simplicial squared volumes
\begin{align}
\mu[\sigma]=\prod_s \sV_s^m
\end{align}
parametrized by $m$. For the Lorentzian case one could use $\mu[\sigma]=\prod_s (-\sV_s)^m$ to make the measure positive definite, in analogy to $\prod_x (-\det g(x))^m$. When the lattice has fixed size this makes no essential difference from (\ref{eq:sqgm1}) since the two measures only differ by an overall constant. This can be included as a term in the integrand exponent
\begin{align}\label{eq:em}
E_m = m \sum_s \log \sV_s.
\end{align}
Any measure factor can be similarly be incorporated by setting $\mu[\sigma]=1$ and introducing an additional term in the integrand exponent. We will adopt this formulation and fix the measure to be
\begin{align}\label{eq:sqgm}
\int_C \mathcal{D}\sigma = (\sum_\tau) \lim_\Gamma \prod_{e\in\Gamma} \int_{-\infty}^\infty d\sigma_e ~  C[\sigma].
\end{align}

The constraint $C[\sigma]$ specifies the integration contour and determines if the theory is for the Euclidean or Lorentzian. It equals $1$ when the Euclidean/Lorentzian generalized triangle inequalities (\ref{eq:egti})/(\ref{eq:lgti}) are matched and vanishes otherwise. In the Lorentzian case, an additional constraint may be imposed so that each point of a simplicial manifold has two lightcones. This is explained in more detail in \cref{sec:lcs}.

On a fixed lattice graph $\Gamma$, the gravitational configurations are summed over by integrating the squared lengths $\sigma_e$ on edges $e$. The continuum limit $\lim_\Gamma$ is taken by going to ever finer lattice graphs (\cref{fig:sll}). In practice, the lattice field theory strategy is usually adopted. Instead of taking the limit, one evaluates the path integral on a fixed graph and look for the continuum limit by searching for universality classes.

Whether topologies should be summed over in the gravitational path integral is an open question \cite{Hartle1985UnrulyGravity}. In (\ref{eq:sqgm}) the sum over topologies $\sum_\tau$ is included as an option enclosed in brackets.

In (\ref{eq:sqg1}) the path integral is expressed in terms of the path exponent $E$ instead of the action $S$ to retain unified formula for the Euclidean, Lorentzian, and general complex cases. $E$ is related to the actions by
\begin{align}
E=\begin{cases}
-S^E, \quad & \text{in Euclidean space,}
\\
iS^L, \quad & \text{in Lorentzian spacetime.}
\end{cases}
\end{align}
Explicitly, $E$ equals
\begin{align}
E=& \underbrace{- \lambda V}_{E_{CC}} + \underbrace{(-k) \sum_h \delta_h \sqrt{\sV_h-0i}}_{E_{EH}}  + \underbrace{\cdots}_{E_O}. \label{eq:pe}
\end{align}
$E_{O}$ stands for ``other terms'' in addition to the cosmological constant term $E_{CC}$ and the Einstein-Hilbert term $E_{EH}$. The measure factor (\ref{eq:em}) is an example. An $R^2$ term as another example is considered in \cref{sec:2dsqg}. The terms $E_{CC}$ and $E_{EH}$ are discussed below.

\subsection{Cosmological constant term}

The cosmological constant term equals
\begin{align}
E_{CC} = & -\lambda V = -\lambda \sum_s  V_s.
\end{align}
Here $\lambda$ is the cosmological constant, and the sum is over all simplicial volumes $V_s=\sqrt{\sV_s}$ as defined in (\ref{eq:vol}).

In Euclidean space $\sV_s>0$, so $V_s>0$. Therefore large volumes are suppressed by the exponent $E_{CC}$. This agrees with ordinary Euclidean quantum gravity. In Lorentzian spacetime $\sV_s<0$, so $V_s=\sqrt{\sV_h}$ as defined in (\ref{eq:vol}) are positive imaginary. This agrees with the usual convention for Lorentzian quantum gravity in which $E_{CC} = - i \lambda V^L$ with a positive Lorentzian volume $V^L=\sum_s \abs{V_s}$.

\subsection{Einstein-Hilbert term}\label{sec:eht}

The Einstein-Hilbert term equals
\begin{align}\label{eq:eht}
E_{EH} = & -k \sum_h \delta_h \sqrt{\sV_h-0i} .
\end{align}
Here $k>0$ is the gravitational coupling constant, the sum is over all hinges $h$, $\sV_h$ is the squared volume of the hinge $h$, and $\delta_h$ is its deficit angle.
The notation $\sqrt{z-0i}$ is as defined in (\ref{eq:bcc}):
\begin{align}
\sqrt{z-0i}=&\sqrt{r}e^{i\phi/2}, \quad z=r e^{i\phi}\text{ with }\phi\in [-\pi,\pi).
\end{align}
The point is that $\sqrt{z-0i}$ is negative imaginary for $z<0$. 

In the Euclidean domain $\sV_h>0$, so $\sqrt{\sV_h-0i}=\sqrt{\sV_h}>0$. In addition, (\ref{eq:la}) agrees with (\ref{eq:ea}). Then $E_{EH} = -k \sum_h \delta_h V_h$ is minus the Einstein-Hilbert term of Euclidean simplicial quantum gravity in the convention of \cite{Hamber2009QuantumApproach}. This in turn yields in the continuum limit
\begin{align}
Z=\int \mathcal{D}g ~ e^{\int d^dx \sqrt{g} (- k R)}
\end{align}
for the pure gravity path integral. Note the extra minus sign in contrast to (\ref{eq:qge}). Since the Einstein-Hilbert term is unbounded from below, it is unclear if this sign choice is a bad one. 
In a follow up work, we will point out a different branch choice for the angle formula (\ref{eq:la}) which reproduces the the Einstein-Hilbert term with the conventional sign in the Euclidean.\footnote{I am very grateful to Bianca Dittrich and Jos{\'e} Padua-Arg{\"u}elles for discussions that clarified the sign conventions of the Einstein-Hilbert term and the mistakes I made regarding the alternatives for the Einstein-Hilbert term in a previous version of the manuscript. The discussions also clarified how one should interpret Sorkin's Lorentzian Regge action \cite{SorkinLorentzianVectors} so that it is holomorphic. The details of this interpretation will be reported elsewhere.}

For a Lorentzian path integral, (\ref{eq:la}) is used to define the deficit angle $\delta_h$ according to (\ref{eq:da2}). We have
\begin{align}\label{eq:eheld1}
E_{EH} = & ik\sum_{h\text{ timelike}} \delta_h \abs{V_h} - k\sum_{h\text{ spacelike}} \delta_h \abs{V_h},
\end{align}
where $\sum_h$ is expanded into a sum over timelike and spacelike hinges (lightlike hinges do not contribute to the exponent since $\sV_h=0$), and $\abs{V_h}$ is the modulus of $V_h=\sqrt{\sV_h}$. 

Sorkin showed that 
\begin{align}
ik\sum_{h\text{ timelike}} \delta_h \abs{V_h} + k\sum_{h\text{ spacelike}} \tilde{\delta_h} \abs{V_h},
\end{align}
reproduces $ik\int d^dx \sqrt{-g} R$ in the continuum limit when $\tilde{\delta_h}$ is positive for a spacelike Lorentz boost deficit angle \cite{Sorkin1974DevelopmentFields}.\footnote{In this statement $R$ is as defined from (\ref{eq:rt}). Note that Sorkin used an opposite sign convention for $R$ in the original paper \cite{Sorkin1974DevelopmentFields}.} By \cref{prop:caba}, in the convention of the present work a spacelike Lorentz boost deficit angle $\delta_h$ is negative imaginary. Therefore (\ref{eq:eheld1}) also reproduces the commonly used path integral exponent $E_{EH}=iS_{EH}=ik\int d^dx \sqrt{-g} R$ of (\ref{eq:qgl}).

\subsection{Lightcone structures}\label{sec:lcs}

In ordinary classical space-time, each point has two lightcones attached to it. In simplicial gravity, a point can have more or fewer than two light cones (\cref{fig:2dlcs}). 

\begin{figure}
    \centering
    \includegraphics[width=.3\textwidth]{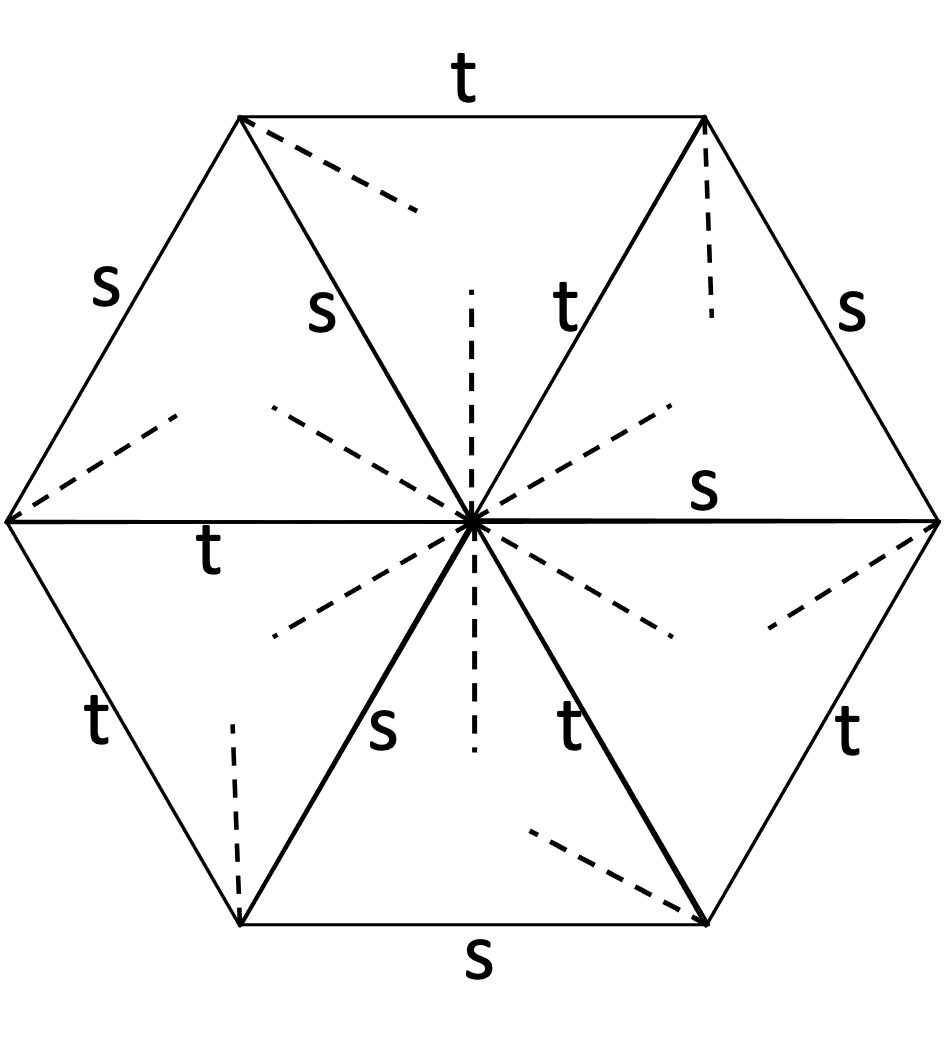}
    \caption{Irregular lightcone structure in $2D$. The point at the center has six light rays (dashed lines) and three lightcones, if spacelike (s) and timelike (t) edges are as assigned.}
    \label{fig:2dlcs}
\end{figure}

It is an open question whether such spacetime configurations with irregular lightcone structures should be included in the gravitational path integral. When they are included the exponent becomes complex rather than staying imaginary. This is because the constant $2\pi$ in the exponents are cancelled exactly when the angles enclose four light rays, as in ordinary flat spacetime (\cref{prop:rtheta}). Depending on the sign choice for the exponent, a space-time configurations with the irregular lightcone structures is either suppressed or enhanced by the additional non-vanishing real part of the exponent. 

In \cite{Louko1995ComplexChange}, reasons are offered to prefer the enhancement (suppression) of configurations with fewer (more) than four light rays. 
The exponent (\ref{eq:eheld1}) with the extra minus sign conforms with the opposite choice. As will be reported in details elsewhere, a different branch choice for the angle formula (\ref{eq:la}) reverses the enhancement/suppression.
If irregular light structures are allowed in Nature, observing the enhancement/suppression effects could in principle help us to determine the branch choice.


\section{Holomorphic flow}\label{sec:hf}

Analytic calculations for the non-perturbatively defined gravitational path integral is hard. In the Euclidean, one usually proceeds numerically with Markov Chain Monte Carlo simulations. The efficiency of this method relies on positivity of the path integrand in the Euclidean. In the Lorentzian, however, the path integrand is complex. The leads to the sign problem. The phase of the complex numbers summed over can fluctuate wildly to cancel each other off, which reduces the efficiency of Markov Chain Monte Carlo simulations. 

The sign problem is not restricted to quantum gravity, but is also encountered in quantum theories of matter. Several methods have been developed to overcome the sign problem (see e.g., \cite{AlexandruComplexProblem, Berger2019ComplexPhysics, Gattringer2016ApproachesTheory} and references therein). The basic idea of the complex path methods is to deform the integration contour to the complex to reduce the phase fluctuations. This idea is demonstrated to work for several models, including low dimension Thirring models, real time scalar field theories, and Hubbard models \cite{AlexandruComplexProblem}. It has also been applied to analyze gravitational propagators for spin-foam models in the large spin limit \cite{Han2021SpinfoamPropagator}.

As reviewed in \cite{AlexandruComplexProblem} there are several different ways to implement the general idea of complex path deformation to overcome the sign problem. In later sections we apply the ``holomorphic gradient flow'' algorithm, also called the ``generalized thimble'' algorithm, \cite{Alexandru2016SignThimbles, Alexandru2017MonteCarloModel} to Lorentzian simplicial quantum gravity. This section summarizes the algorithm.

\subsection{Flow equations}\label{sec:fe}

The celebrated Cauchy integration theorem indicates that up to a sign the integral of a complex function $f(z)$ does not change value if the integration contour is deformed through a region where $f(z)$ is holomorphic. 

Cauchy's theorem admits a multi-dimensional generalization \cite{AlexandruComplexProblem} which applies to path integrals of multiple variables. The holomorphic gradient flow algorithm exploits this to find deformed contours where the sign problem is mitigated. Consider a path integral with a holomorphic integrand of the form
\begin{align}\label{eq:pi1}
Z =& \int \mathcal{D}\sigma ~ e^{E[\sigma]},
\end{align}
where in $\mathcal{D}\sigma$ multiple configurations $\sigma_e$ labelled by the lattice edges $e$ are integrated over. The \textbf{flow equations} are
\begin{align}\label{eq:fe}
\frac{d \sigma_{e}}{dt}=&-\overline{\partial_e E} \quad \forall e,
\end{align}
where $\partial_{e}$ is a shorthand for $\pdv{}{\sigma_{e}}$, and the overline stands for complex conjugation. For any point $\zeta$ in the original integration contour, the solution to (\ref{eq:fe}) as a function of the flow time $t$ defines the \textbf{holomorphic gradient flow} (or holomorphic flow in short) for $\zeta$. Solving (\ref{eq:fe}) for the whole original integration contour yields a deformation of the integration contour as a function $t$. 

\begin{figure}
    \centering
    \includegraphics[width=.6\textwidth]{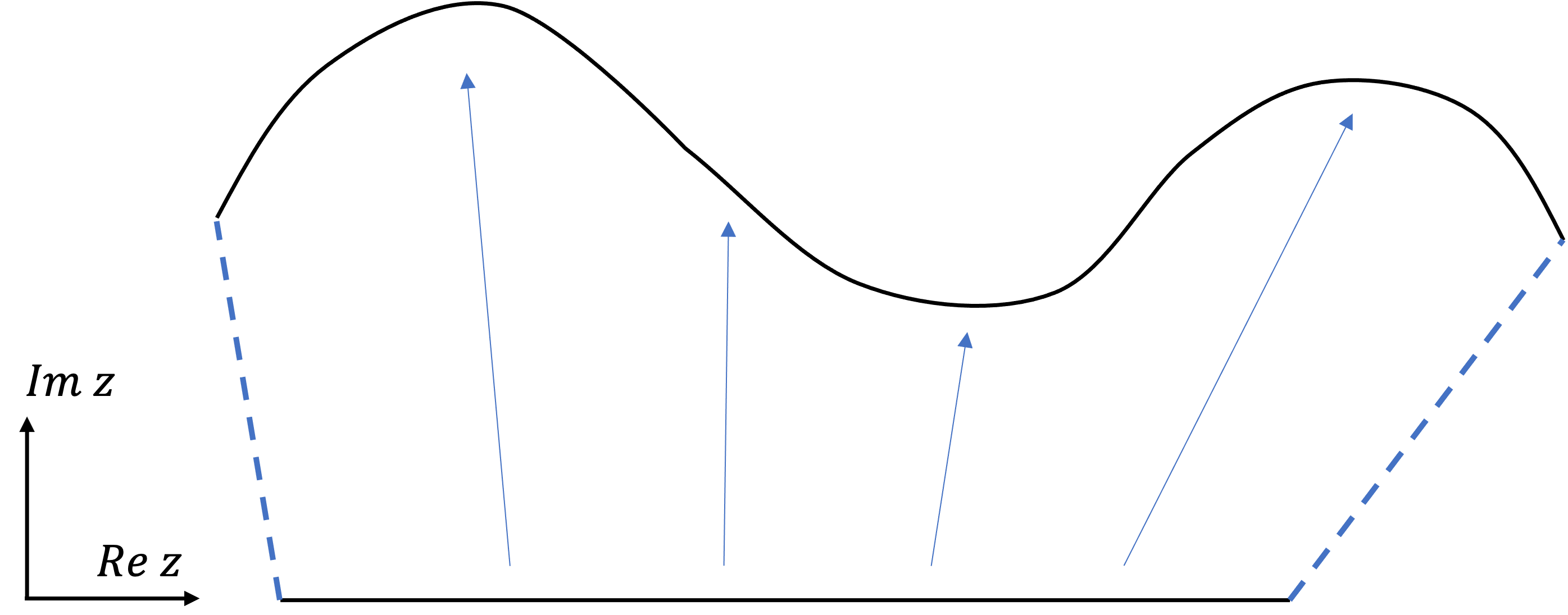}
    \caption{Schematic illustration of the flow region and its boundary. The original contour at the bottom is deformed into the contour at the top. The integral along these contours plus on the dashed boundaries is zero, if the function being integrated over is holomorphic inside. If the integral on the dashed boundaries are negligibly small, then the integrals on the two contours are equal up to a sign.}
    \label{fig:hfcb}
\end{figure}

If the integral along the boundary of the flowed region is negligible, then up to a sign (\ref{eq:pi1}) can be evaluated on the flowed contour (\cref{fig:hfcb}). This could reduce the phase fluctuations for the complex numbers integrated over, because only a smaller region on the flowed contour contribute significantly to the integral, and the phase fluctuations could be small in this smaller region. 

To see this, we look at the real part $E_R$ and the imaginary part $E_I$ of $E$. By (\ref{eq:fe}),
\begin{align}\label{eq:DreEDt}
\dv{E_R}{t}=&\frac{1}{2}(\dv{E}{t}+\overline{\dv{E}{t}})=\frac{1}{2}\sum_e (\partial_e E \dv{\sigma_e}{t}+\overline{\partial_e E \dv{\sigma_e}{t}})=-\sum_e\abs{\partial_e E}^2\le 0,
\\\dv{E_I}{t}=&\frac{1}{2i}(\dv{E}{t}-\overline{\dv{E}{t}})=\frac{1}{2i}\sum_e (\partial_e E \dv{\sigma_e}{t}-\overline{\partial_e E \dv{\sigma_e}{t}})= 0.
\end{align}
Therefore the real part of the exponent decreases monotonically through the flow, while the imaginary part stays constant. 
For sufficiently long flow time, the magnitude of the integrand is exponentially suppressed for most points on the deformed contour. Only points close to the critical points of the flow obeying 
\begin{align}
\partial_e E=0 \quad \forall e
\end{align}
contribute significantly. 

If the phase fluctuations for such points that contribute significantly is small enough, Markov Chain Monte Carlo simulation can be efficiently performed. 

\subsection{Numerical algorithm}\label{sec:na}

As a summary of \cref{sec:fe}, suppose:
\begin{itemize}
\item The holomorphic flow transverse a region where the path integrand is holomorphic;
\item The boundary of the flow region have negligible contribution to the path integral.
\end{itemize}
Then the original path integral can be equally evaluated along the contour at any flow time $t=T$.

To compute the path integral on the flowed contour, one could use the holomorphic gradient flow algorithm \cite{Alexandru2016SignThimbles, Alexandru2017MonteCarloModel}. The idea is to parametrize the flowed contour by its preimage in the original contour, and perform Markov Chain Monte Carlo simulation using weights on the flowed contour. Specifically, the algorithm goes as:
\begin{enumerate}
\item Start with a configuration $\zeta$ in the original contour. Evolve it under the holomorphic flow by time $T$ to obtain $\phi=\phi(\zeta)$.
\item Draw a new configuration $\zeta'=\zeta+\delta\zeta$ on the original contour, where $\delta\zeta$ is a random vector drawn from a symmetric distribution. Again evolve $\zeta'$ under the flow by time $T$ to obtain $\phi'=\phi'(\zeta')$.
\item Accept $\zeta'$ with probability $P = \min\{1, e^{\Re E_{\text{eff}}(\phi')-\Re E_{\text{eff}}(\phi)} \}$, where $E_{\text{eff}}$ is defined in (\ref{eq:Eeff}).
\item Repeat steps 2 and 3 until a sufficient ensemble of configurations is generated.
\item Compute the expectation values using
\begin{align}\label{eq:expo}
\ev{O}=&\frac{\ev{O e^{i\varphi(\zeta)}}_{\Re E_{\text{eff}}}}{\ev{e^{i\varphi(\zeta)}}_{\Re E_{\text{eff}}}},
\end{align}
where $\ev{\cdot}_{\Re E_{\text{eff}}}$ stands for the average using the ensemble just generated, and $\varphi$ is defined in (\ref{eq:vphi}).
\end{enumerate}

In steps 1 and 2, the evolution can be conducted through numerically integrating the ODEs (\ref{eq:fe}). If the complexified theory has is domain on Riemann surfaces, as is the case for simplicial quantum gravity, branches need to be recorded as part of the numerical integration algorithm to make sure the system flows continuously on the Riemann surfaces. In Step 3,
\begin{align}\label{eq:Eeff}
E_{\text{eff}}(\phi) =& E(\phi(\zeta)) + \log \det J(\zeta),
\quad
J_{ee'} = \pdv{\phi_e}{\zeta_{e'}},
\end{align}
where $\phi_e$ and $\zeta_{e}$ are the values $\phi$ and $\zeta$ take on the edge $e$. The Jacobian can be obtained (see Appendix A of \cite{AlexandruComplexProblem}) by integrating
\begin{align}\label{eq:jcb}
\frac{d J_{ee'}}{dt}= \sum_{e''}\overline{H_{ee''}J_{e''e'}},\quad H_{ee'}:=-  \partial_{e'}\partial_e E, \quad J_{ee'}(0)=\delta_{ee'}.
\end{align}
The function $e^{E_{\text{eff}}}$ is the integrand of the final integral to be computed, since
\begin{align}
Z=&\int_{M_0} e^{E(\zeta)} d\zeta 
\\=& \int_{M_T} e^{E(\phi)} d\phi
\\=& \int_{M_0} e^{E(\phi(\zeta))} \det J ~d\zeta,
\label{eq:ppi}
\end{align}
where we reparametrized the flowed manifold $M_T$ by points $\zeta$ of the original manifold $M_0$ in the last step. Now the integrand equals $e^{E_{\text{eff}}}$ for $E_{\text{eff}}$ defined in (\ref{eq:Eeff}). Expanding $E_{\text{eff}}$ in real and imaginary parts yields $e^{E_{\text{eff}}}= e^{\Re E_{\text{eff}}  + i \varphi}$, where
\begin{align}\label{eq:vphi}
\varphi= \Im E_{\text{eff}}=\Im E + \arg\det(J).
\end{align}
This explains steps 3 and 5, in which we sample (\ref{eq:ppi}) according to the magnitude $e^{\Re E_{\text{eff}}}$ of the integrand, and treat the phase $e^{i\varphi}$ as part of the observable in (\ref{eq:expo}).

This algorithm can alleviate the sign problem because as $T\rightarrow\infty$, the flowed manifold approaches a combination of steepest descent contours (Lefschetz thimbles) on each of which $\varphi$ is constant \cite{AlexandruComplexProblem}.

However, the usefulness of the algorithm is not guaranteed because of ``trapping'' for the Monte Carlo sampling. As noted below (\ref{eq:DreEDt}) $\Re E$ decreases monotonically under the holomorphic flow, so $\Re E_{\text{eff}}$ also tends to decrease. As $T$ is increased, the probability weight $e^{\Re E_{\text{eff}}}$ develop peaks around the stationary points where $\partial_e E=0$, separated by valleys where $e^{\Re E_{\text{eff}}}$ is exponentially suppressed. Consequently it can be hard for the Markov chain to travel across the peak regions to generate a sufficient sample.

In practice, we need to find a flow time $T$ large enough so that the phase fluctuation in $\varphi$ is sufficiently suppressed to tame the sign problem, and small enough so that the trapping of the Markov chain is sufficiently weak. More sophisticated algorithms such as the tempering algorithms \cite{Fukuma2017ParallelThimbles, Alexandru2017TemperedThimbles} involving multiple flow times/chains have been developed to avoid the trapping issue. In principle general Markov Chain Monte Carlo algorithms for multimodal distributions can also be applied.

\section{2D simplicial quantum gravity}\label{sec:2dsqg}

We apply the holomorphic gradient flow method to overcome the sign problem for Lorentzian simplicial gravitational path integrals. We focus on the 2D case for this initial study on the topic. The relevant expressions for the holomorphic flow equation and the Jacobian equation are given in this section. Along the way we prove a complex version of the Gauss-Bonnet theorem, which may be of independent interest. The numerical results are presented in the next section.

In $2D$, we consider the path integral
\begin{align}
Z =& \int \mathcal{D}\sigma ~ e^E,
\\
E = & - \lambda V-k \sum_v \delta_v + a \sum_v \frac{\delta_v^2}{A_v} + m \sum_t \log \sV_t.
\end{align}
The first (cosmological constant) and second (Einstein-Hilbert) terms are as is (\ref{eq:pe}) specialized to $2D$. The fourth term is the measure factor term of (\ref{eq:em}). The third term $a \sum_v \delta_v^2/A_v$ is the $R^2$ term \cite{Hamber1986Two-dimensionalGravity}. Here $a$ is the coupling constant, and $A_v$ is the area share of vertex $v$:
\begin{align}\label{eq:av}
A_v =& \frac{1}{3}\sum_{t\ni v} V_t = \frac{1}{3}\sum_{t\ni v} \sqrt{\sV_t}, 
\end{align}
where the sum is over triangles $t$ containing vertex $v$, and $\sV_t$ is the squared volume for triangle $t$ calculated according to (\ref{eq:svol1}) or (\ref{eq:svol}). The letter $A$ instead of $V$ is used for $A_v$ to distinguish from the hinge (vertex in $2D$) volume $V_h=V_v$, which is usually set to $1$ in $2D$. 



\subsection{Complex Gauss-Bonnet theorem}

The Einstein-Hilbert term $E_{EH}=-k \sum_v \delta_v$ can actually be left out of the path integration because it is topological.

In the Euclidean domain, the celebrated Gauss-Bonnet theorem says that $E_{EH}=k 2\pi \chi$, where $\chi$ is a topological invariant that is fixed by the simplicial complex, and does not depend on the particular length assignments. The same holds in the Lorentzian domain. A nice prove can be found in \cite{SorkinLorentzianVectors}, and a slight generalization that accounts for multiple boundary components can be found in \cite{Jia2022Time-spaceGravity}.

That a version of the Gauss-Bonnet theorem exists in the complex domain was suggested by Louko and Sorkin \cite{Louko1995ComplexChange}, but they left it as an open question to investigate. 

Here we prove a complex version of the Gauss-Bonnet theorem, which generalizes the Euclidean and Lorentzian versions. It implies that on a fixed simplicial lattice, $E_{EH}$ is constant when the Lorentzian or Euclidean contour is continuously deformed into the complex domain. Therefore $E_{EH}$ can be taken out of the path integral in the holomorphic gradient flow algorithm.
\begin{theorem}[Complex Gauss-Bonnet]\label{th:cgb}
On a fixed simplicial lattice, any continuous deformation of the path integration contour in the complex domain will not change the value of the Einstein-Hiblert term $E_{EH}$. 

If the deformation is continuously connected to the Lorentzian or the Euclidean contour,
\begin{align}\label{eq:cgbt1}
E_{EH}/(-k) = 2\pi \chi, \quad \chi=V-E+T,
\end{align}
where $V,E,T$ are the vertex, edge, and triangle numbers of the simplicial lattice, and $\chi$ is Euler number. This simple result assumes that each boundary vertex is shared by two regions. 

More generally, when the numbers of regions sharing the vertices $v$ is $Q_v$,
\begin{align}\label{eq:cgbt2}
E_{EH}/(-k) = 2\pi \chi, \quad \chi= V^{\mathrm{o}} + \frac{1}{2} V^\partial - E + T + \sum_{v\in \partial} \frac{1}{Q_v},
\end{align}
where the bulk and boundary elements are labelled by superscripts $\mathrm{o}$ and $\partial$, and the sum $\sum_{v\in \partial}$ is over all boundary vertices.
\end{theorem}
\begin{proof}
In $2D$, the Einstein-Hilbert equals
\begin{align}
E_{EH}/(-k)=\sum_v \delta_v=&(\sum_v 2\pi/Q_v - \sum_a \theta_a)
\\
=& (\sum_v 2\pi/Q_v - \pi N).\label{eq:cgb1}
\end{align}
In the first line we used the definition (\ref{eq:da2}) of the deficit angle. In $\delta_v$ for each vertex $v$, there is a sum over angles $\theta$ around that vertex. After $\sum_v$, we obtain a sum $\sum_a \theta_a$ is over all triangular angles of the $2D$ simplicial complex. In the second line we grouped the angles into triangles and applied \cref{th:sta}. Here $N$ is some integer. This shows that the $E_{EH}$ can only take values from a discrete set labelled by $N$.

Under a continuous deformation of the contour, a holomorphic function such as $E_{EH}$ can only change value continuously. Yet we just showed that the codomain of $E_{EH}$ is a discrete set. Therefore $E_{EH}$ cannot change value under a continuous deformation of the contour. 

The claims (\ref{eq:cgbt1}) and (\ref{eq:cgbt2}) can be proved by the same argument in \cite{SorkinLorentzianVectors} and \cite{Jia2022Time-spaceGravity}. In the Lorentzian and Euclidean domains,
\begin{align}\label{eq:gb1}
E_{EH}/(-k\pi)=&2V^{\mathrm{o}} + \sum_{v\in \partial} \frac{2}{n_v} - T,
\\
0=&-2 E^{\mathrm{o}} - E^\partial + 3T,
\label{eq:gb2}
\\
0=&V^{\partial} - E^\partial.
\label{eq:gb3}
\end{align}
Equation (\ref{eq:gb1}) uses the fact that in the interior of the region, $Q_v=1$, and that in the Lorentzian and Euclidean domains the angles of a triangle sum to $\pi$ (\cref{prop:ltri}), whence $N=T$. Equations (\ref{eq:gb2}) and (\ref{eq:gb3}) are simple facts about the simplicial lattice. Each bulk edge is shared by two faces, each boundary edge is shared by one face, and each face has three edges so (\ref{eq:gb2}) follows. The boundary is formed by a vertex-edge-vertex-edge... chain so (\ref{eq:gb3}) follows. Adding up (\ref{eq:gb1}) to (\ref{eq:gb3}) yields (\ref{eq:cgbt2}). Specializing to $Q_v=2$ for all $v$ yields (\ref{eq:cgbt1}).
\end{proof}

\subsection{Flow equations}

Because of \cref{th:cgb}, $\partial_e E_{EH}=0$, so the flow equations (\ref{eq:fe}) become
\begin{align}
\frac{d \sigma_{e}}{dt}=&-\overline{\partial_e E}=-\overline{\partial_e E_{CC}}-\overline{\partial_e E_{R^2}}-\overline{\partial_e E_{m}}.
\end{align}
For the cosmological constant term $E_{CC}$,
\begin{align}
\partial_e E_{CC}
=&  -\lambda \partial_e V
\\
=& -\lambda \sum_t \partial_e V_t.
\end{align}
For the $R^2$ term $E_{R^2}$,
\begin{align}
\partial_{e} E_{R^2}
=&  a \sum_v \partial_{e}(\frac{\delta_v^2}{A_v})
\\=& a \sum_v [\frac{2\delta_v \partial_{e}\delta_v}{A_v}-\frac{\delta_v^2 \partial_{e}A_v}{A_v^2}].
\end{align}
For the measure term $E_m$,
\begin{align}
\partial_{e} E_{m}
=&  m \sum_t \partial_{e} \log \sV_t
\\=&  m \sum_t \sV_t^{-1}\partial_{e} \sV_t.
\end{align}
Therefore
\begin{align}
\frac{d \sigma_{e}}{dt}=& 
-\overline{\partial_e E_{CC}}-\overline{\partial_e E_{R^2}}-\overline{\partial_e E_{m}}
\\
=& \lambda \sum_t \overline{\partial_e V_t}
- a \sum_v \overline{(\frac{2\delta_v \partial_{e}\delta_v}{A_v}-\frac{\delta_v^2 \partial_{e}A_v}{A_v^2})}-m \sum_t\overline{\sV_t^{-1}\partial_{e} \sV_t}.\label{eq:2dfe}
\end{align}
This formula needs to be expressed in terms of the squared lengths to be applied. While $\delta_v$, $A_v$, and $\sV_t$ in terms of the squared lengths are known from the definitions, the derivative terms in terms of the squared lengths are given below.

\subsection*{Volume terms}

For $\partial_e V_t$ and $\partial_e A_v$, a straightforward calculation using the definitions yields
\begin{align}
\partial_e V_t=&\frac{\partial_e \sV_t}{2 \sqrt{\sV_t}}=\frac{\partial_e \sV_t}{2 V_t},
\label{eq:2dvbe}
\\\partial_e \sV_t=&\frac{1}{8} \left(-\sigma _e+\sigma _{e1}+\sigma _{e2}\right),  \label{eq:2dv2be}
\\\partial_e A_v =& \frac{1}{3}\sum_{t\ni v} \partial_e V_t=\frac{1}{3}\sum_{t\ni v, e} \partial_e V_t,
\end{align}
where $e1, e2$ are the other two edges of the triangle $t$. 

\subsection*{Angle terms}

For $\partial_{e}\delta_v$,
\begin{align}
\delta_v =& 2\pi/Q_v - \sum_{t\ni v}\theta_{t,v},
\\\partial_{e}\delta_v =& -\sum_{t\ni v} \partial_{e} \theta_{t,v} = -\sum_{t\ni v, e} \partial_{e} \theta_{t,v}.\label{eq:dabl2}
\end{align}
For $a$ and $b$ in triangle $t$ meeting at vertex $v$, (\ref{eq:ca}) implies
\begin{align}
\frac{\partial \theta _{t,v}}{\partial \sigma_{a}} = & \frac{\sigma _a-\sigma _b+\sigma _c}{4i \sigma _a (a\wedge b) }=\frac{\sigma _a-\sigma _b+\sigma _c}{8 \sigma _a V_t },
\label{eq:abe1}
\\
\frac{\partial \theta _{t,v}}{\partial \sigma_{b}} = & \frac{-\sigma _a+\sigma _b+\sigma _c}{4i \sigma _b  (a\wedge b)}=\frac{-\sigma _a+\sigma _b+\sigma _c}{8 \sigma _b V_t},
\label{eq:abe2}
\\
\frac{\partial \theta _{t,v}}{\partial \sigma_{v}} = & \frac{i}{2 a\wedge b} = \frac{-1}{4 V_t}.
\label{eq:abe3}
\end{align}
Here we noted that
\begin{align}\label{eq:awedgeb2}
a\wedge b=& -2 i V_t,
\end{align}
where $V_t$ in terms of squared lengths is given in (\ref{eq:2dvol}). These can be used to express (\ref{eq:dabl2}) fully in the squared lengths.

\subsection{Jacobian}

The Jacobian flow equation is given in (\ref{eq:jcb}) as 
\begin{align}
\frac{d J_{ee'}}{dt}= \sum_{e''}\overline{H_{ee''}J_{e''e'}},\quad H_{ee'}:=-  \partial_{e'}\partial_e E, \quad J_{ee'}(0)=\delta_{ee'}.
\end{align}
Specialized to simplicial quantum gravity in $2D$,
\begin{align}
H_{ee'}=& -  \partial_{e'}\partial_e E = -  \partial_{e'}\partial_e E_{CC}-  \partial_{e'}\partial_e E_{R^2}-\partial_{e'}\partial_e E_{m},
\end{align}
where the Einstein-Hilbert term drop out by \cref{th:cgb}.

\subsection*{The cosmological constant term}

The cosmological constant term is
\begin{align}
\partial_{e'}\partial_e E_{CC} =& -\lambda \sum_t\partial_{e'}\partial_e V_t
\\ =& -\lambda \sum_{t\ni e, e'} \partial_{e'}\partial_e V_t,
\label{eq:2dECCbee}
\end{align}
where it was noted that $\partial_{e'}\partial_e V_t= 0$ if the triangle $t$ does not contain both $e$ and $e'$. By (\ref{eq:2dvbe}) and (\ref{eq:2dv2be}),
\begin{align}
\partial_{e'}\partial_e V_t=&\frac{1}{2 \sqrt{\sV_t}} ( \frac{-1}{2 \sV_t} \partial_e \sV_t \partial_{e'} \sV_t + \partial_{e'} \partial_e \sV_t).
\label{eq:2dvbee}
\\\partial_e \sV_t=&\frac{1}{8} \left(-\sigma _e+\sigma _{e1}+\sigma _{e2}\right),
\\\partial_{e'} \partial_e \sV_t=&
\begin{cases}
\frac{-1}{8}, \quad & e = e',
\\\frac{1}{8}, & e \ne e'.
\end{cases}
\label{eq:2dv2bee}
\end{align}
Plugging these in (\ref{eq:2dECCbee}) yields an expression in terms of squared lengths.

Regarding computational complexity it is relevant to note that $\partial_{e'}\partial_e E_{CC}$ is quasi-local. Because the sum $\sum_{t\ni e, e'}$ in (\ref{eq:2dECCbee}) is over triangles $t$ that contain both $e$ and $e'$, if $e$ and $e'$ are not identical or adjacent then $\partial_{e'}\partial_e E_{CC}=0$. 

\subsection*{The $R^2$ term}
For the $R^2$ term,
\begin{align}\label{eq:2dER2bee}
\partial_{e} \partial_{e'} E_{R^2}
=& \sum_v \frac{a }{A_v^3}  [2 \delta_vA_v\left(-\delta_v^{(0,1)} A_v^{(1,0)}-\delta_v^{(1,0)} A_v^{(0,1)}+\delta_v^{(1,1)} A_v\right)\nonumber
\\
&+\delta_v^2 \left(2 A_v^{(0,1)} A_v^{(1,0)}-A_vA_v^{(1,1)}\right)+2 \delta_v^{(0,1)} \delta_v^{(1,0)} A_v^2],
\end{align}
where $f^{(i,j)}$ is the shorthand for $\partial_{e}^i \partial_{e'}^j f$.

\begin{figure}
    \centering
    \includegraphics[width=.4\textwidth]{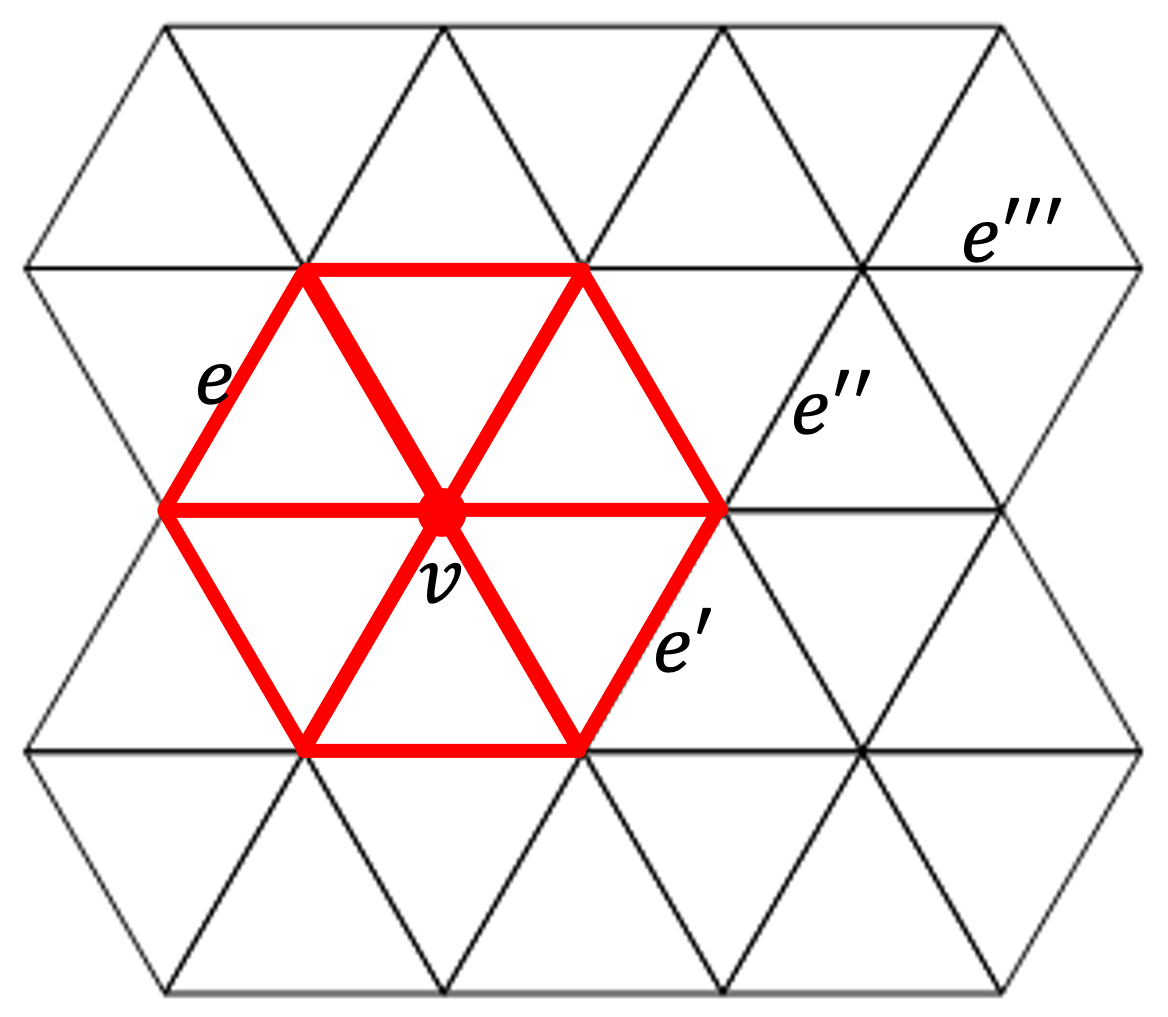}
    \caption{The edges that $A_v$ and $\delta_v$ depend on are thickened. They are all within one edge away from $v$, and are all within two edges away from each other. A pair of edges (e.g., $e$ and $e'''$) more than two edges away will not find any vertex $v$ whose $A_v$ and $\delta_v$ depend on them both. Even a pair of edges (e.g., $e$ and $e''$) two edges away may not find any vertex $v$ whose $A_v$ and $\delta_v$ depend on them both.}
    \label{fig:ql2d}
\end{figure} 

We see that $\partial_{e} \partial_{e'} E_{R^2}$ is quasi-local, in the sense that $\partial_{e} \partial_{e'} E_{R^2}=0$ when $e$ and $e'$ are more than two edges away (meaning the shortest lattice graph path touching both $e$ and $e'$ has more than two edges) (\cref{fig:ql2d}). This is because $\partial_e\delta_v=\partial_e A_v=0$ if $e$ is more than one edge away from $v$. If $e$ and $e'$ are more than two edges away, then at least one of them is more than one edge away from $v$ for any $v$, whence all terms on the right hand side of (\ref{eq:2dER2bee}) vanish.

\subsubsection*{Volume terms}

By the definition of $A_v$,
\begin{align}
A_v^{(1,0)}=&\partial_e V_v = \frac{1}{3}\sum_{t\ni v} \partial_e  V_t=\frac{1}{3}\sum_{t\ni v,e,e'} \partial_e V_t,
\\
A_v^{(1,0)}=&\partial_{e'} V_v = \frac{1}{3}\sum_{t\ni v} \partial_{e'} V_t=\frac{1}{3}\sum_{t\ni v,e,e'}  \partial_{e'} V_t,
\\
A_v^{(1,1)}=&\partial_e \partial_{e'} V_v = \frac{1}{3}\sum_{t\ni v} \partial_e \partial_{e'} V_t=\frac{1}{3}\sum_{t\ni v,e,e'} \partial_e \partial_{e'} V_t.
\end{align}
These can be expressed in terms of squared lengths using (\ref{eq:2dvbe}), (\ref{eq:2dv2be}), (\ref{eq:2dvbee}), and (\ref{eq:2dv2bee}):
\begin{align}
\partial_e V_t =&  \frac{\partial_e \sV_t}{2 \sqrt{\sV_t}},
\\
\partial_e \sV_t=&\frac{1}{8} \left(-\sigma _e+\sigma _{e1}+\sigma _{e2}\right),
\\
\partial_{e'}\partial_e V_t =&\frac{1}{2 \sqrt{\sV_t}} ( \frac{-1}{2 \sV_t} \partial_e \sV_t \partial_{e'} \sV_t + \partial_{e'} \partial_e \sV_t),
\\\partial_{e'} \partial_e \sV_t=&
\begin{cases}
\frac{-1}{8}, \quad & e = e',
\\\frac{1}{8}, & e \ne e'.
\end{cases}
\end{align}

\subsubsection*{Angle terms}

\begin{figure}
    \centering
    \includegraphics[width=.4\textwidth]{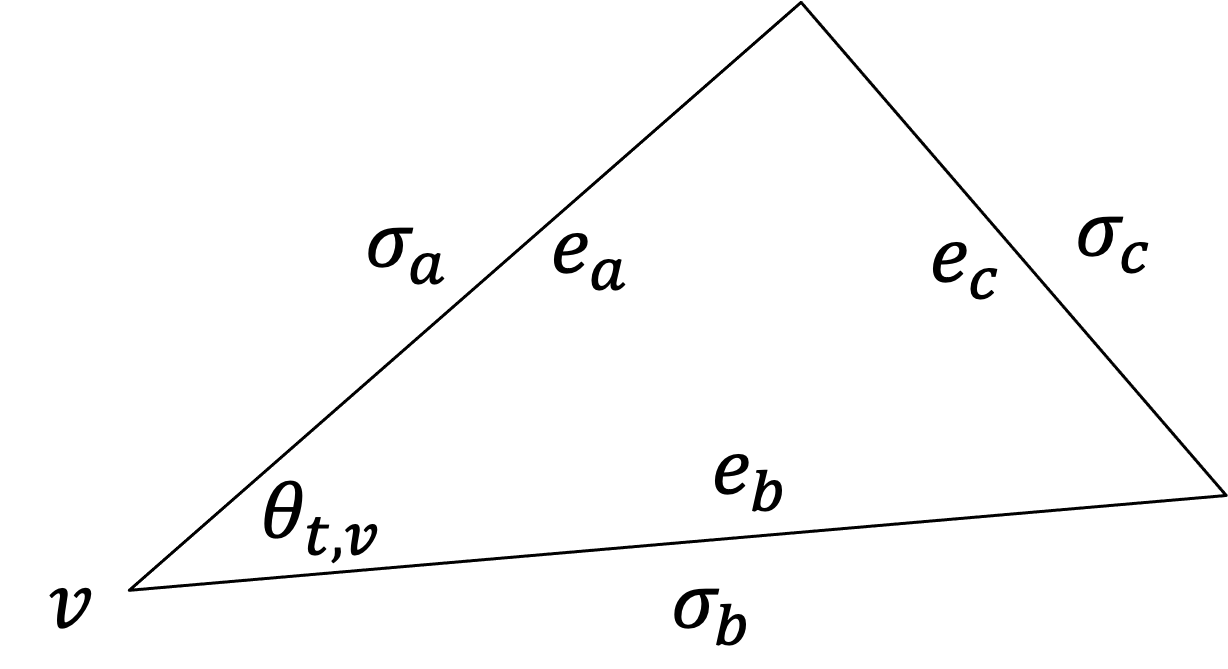}
    \caption{Triangle $t$ with edges $e_a, e_b, e_c$ whose squared lengths are $\sigma_a, \sigma_b, \sigma_c$. Edges $e_a$ and $e_b$ bound the angle $\theta_{t,v}$.}
    \label{fig:tri2}
\end{figure}

The terms $\delta_v^{(1,0)}$ and $\delta_v^{(0,1)}$ can be expressed in squared lengths using (\ref{eq:dabl2}) - (\ref{eq:abe3}) (with labels specified in \cref{fig:tri2}): 
\begin{align}
\partial_{e}\delta_v =& -\sum_{t\ni v} \partial_{e} \theta_{t,v} = -\sum_{t\ni v, e} \partial_{e} \theta_{t,v}.
\\\frac{\partial \theta _{t,v}}{\partial \sigma_{a}} = & \frac{\sigma _a-\sigma _b+\sigma _c}{4i \sigma _a (a\wedge b) }=\frac{\sigma _a-\sigma _b+\sigma _c}{8 \sigma _a V_t },
\\
\frac{\partial \theta _{t,v}}{\partial \sigma_{b}} = & \frac{-\sigma _a+\sigma _b+\sigma _c}{4i \sigma _b  (a\wedge b)}=\frac{-\sigma _a+\sigma _b+\sigma _c}{8 \sigma _b V_t},
\\
\frac{\partial \theta _{t,v}}{\partial \sigma_{v}} = & \frac{i}{2 a\wedge b} = \frac{-1}{4 V_t}.
\end{align}

For the second derivative,
\begin{align}
\delta_v^{(1,1)}=\partial_{e}\partial_{e'}\delta_v =& -\sum_{t\ni v, e, e'} \partial_{e}\partial_{e'} \theta_{t,v}.
\end{align}
For $e,e'$ ordered as $e_a,e_b,e_c$ (\cref{fig:tri2}), the Hessian matrix is
\begin{align}
\partial_{e}\partial_{e'}\theta_{t,v} = \frac{1}{32 V_t^3} \left(
\begin{array}{ccc}
\frac{  X}{4 \sigma _a^2  } & -\sigma _c & \frac{  \left(-\sigma _a+\sigma _b+\sigma _c\right)}{2} \\
 -\sigma _c & \frac{ Y}{4 \sigma _b^2  } & \frac{  \left(\sigma _a-\sigma _b+\sigma _c\right)}{2} \\
 \frac{  \left(-\sigma _a+\sigma _b+\sigma _c\right)}{2} & \frac{  \left(\sigma _a-\sigma _b+\sigma _c\right)}{2} & \frac{  \left(\sigma _a+\sigma _b-\sigma _c\right)}{2} \\
\end{array}
\right),
\end{align}
where 
\begin{align}
   X=&\sigma _a^3+\sigma _a^2 \left(\sigma _c-3 \sigma _b\right)+3 \sigma _a \left(\sigma _b^2-\sigma _c^2\right)-\left(\sigma _b-\sigma _c\right){}^3,
   \\
   Y=X(\sigma _a \leftrightarrow \sigma _b)=&\sigma _b^3+\sigma _b^2 \left(\sigma _c-3 \sigma _a\right)+3\sigma _b \left( \sigma _a^2- \sigma _c^2\right)-\left(\sigma _a-\sigma _c\right){}^3.
\end{align}
The above volume and angular terms of derivatives can be plugged into (\ref{eq:2dER2bee}) to express it in terms of squared lengths.

\subsection*{The measure term}

By the definition of $E_m$,
\begin{align}\label{eq:2dEmbee}
\partial_{e'}\partial_e E_{m}
=& m \sum_{t\ni e, e'}\partial_{e'}\partial_e \log \sV_t
\\=& m \sum_{t\ni e, e'} \frac{1}{\sV_t^2}(\sV_t \partial_{e'}\partial_e \sV_t- \partial_e \sV_t \partial_{e'} \sV_t).
\end{align}
The previous formulas (\ref{eq:2dv2be}) and (\ref{eq:2dv2bee}) can then be used to express this in terms of squared length.

\section{Numerical results}\label{sec:nr}

In this section we present results of numerical simulation for the path integral
\begin{align}
Z =& \int \mathcal{D}\sigma ~ e^E, \quad E = - \lambda V + a \sum_v \frac{\delta_v^2}{A_v} + m \sum_t \log \sV_t,
\end{align}
parameterized by $p=(\lambda, a, m)$. These constants and the squared lengths are set unitless in this section for simplicity.

We compute the expectation value for the squared length $\ev{\sigma_e}=\int \mathcal{D}\sigma ~ \sigma_e e^E$. According to (\ref{eq:expo}),
\begin{align}
\ev{\sigma_e}=&\frac{\ev{\sigma_e e^{i\varphi }}_{\Re E_{\text{eff}}}}{\ev{e^{i\varphi }}_{\Re E_{\text{eff}}}},
\end{align}
where $\ev{\cdot}_{\Re E_{\text{eff}}}$ is the average using the ensemble just generated, and the phase $\varphi$ is the imaginary part of $E_{\text{eff}}$. 

When $\varphi$ fluctuates wildly, the sign problem is bad. The task is to choose $T$ so that on the flowed contour the phase fluctuation is reduced. We can quantify the performance of the algorithm in alleviating the sign problem by the average phase
\begin{align}
\Phi=\abs{\ev{e^{i\varphi }}_{\Re E_{\text{eff}}}}=\abs{\frac{\int \mathcal{D}\sigma ~ e^{i\varphi+ \Re E_{\text{eff}}}}{\int \mathcal{D}\sigma ~ e^{\Re E_{\text{eff}}}}}.
\end{align}
The closer $\Phi$ is to $1$, the less the sign fluctuation, and hence the better the performance. 

In the cases considered below complex contours are found where $\Phi>0.9$.

\subsection{Numerical setup}

\begin{figure}
    \centering
    \includegraphics[width=.3\textwidth]{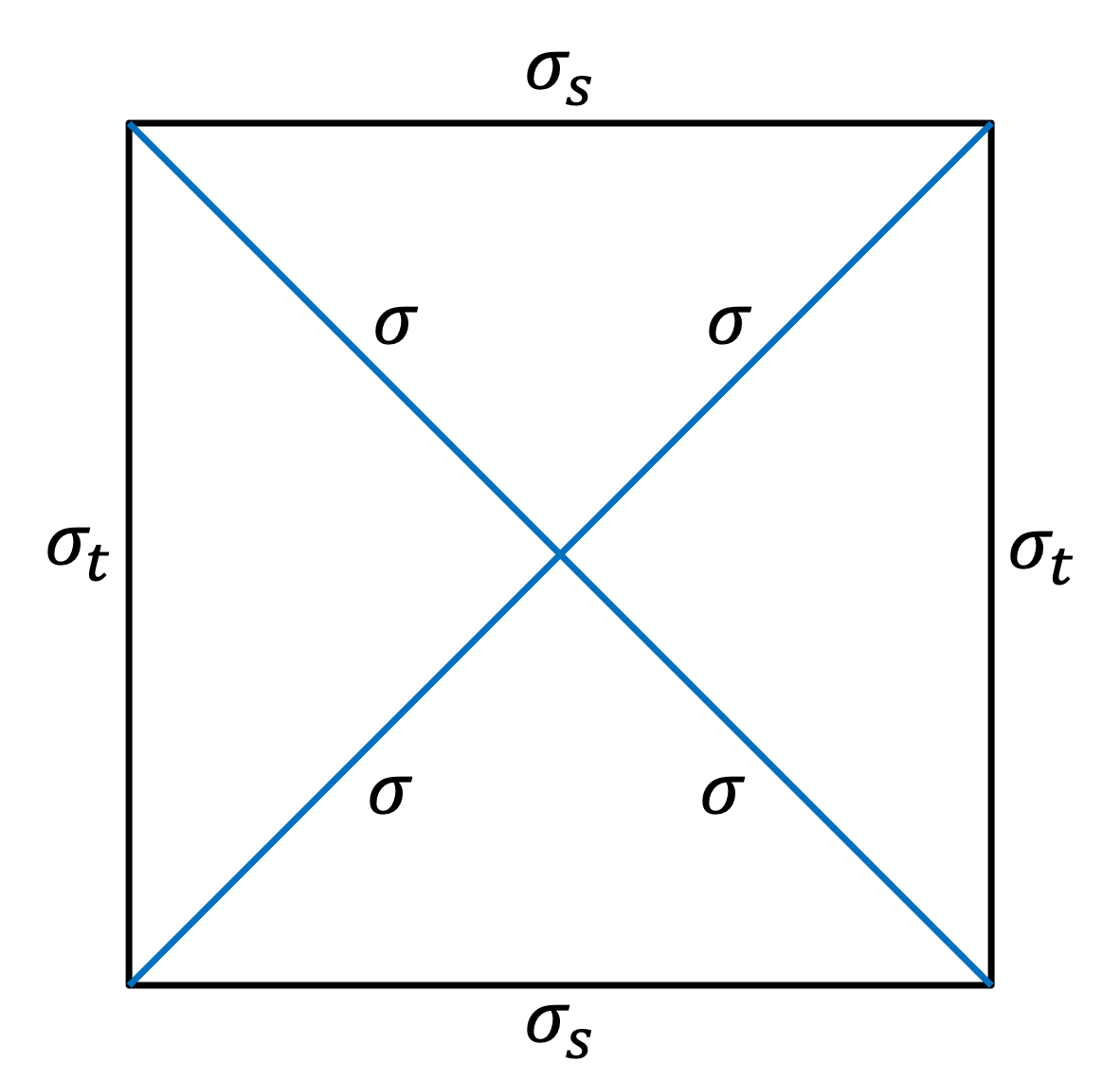}
    \caption{The symmetry-reduced box model with boundary squared lengths $\sigma_t, \sigma_s$ fixed, and interior squared length $\sigma$ dynamical.}
    \label{fig:sqg_box}
\end{figure}

The numerical simulation is performed on a simple box model in a symmetry-reduced setting (\cref{fig:sqg_box}). The boundary squared lengths are fixed at 
\begin{align}
    \sigma_s=1.0, \sigma_t=-1.0.
\end{align}
The four remaining edges are dynamical, and they take the same $\sigma$. In the definition of the deficit angle (\ref{eq:da2}) we take $Q=1$ for the interior vertex and $Q=4$ for the boundary vertices so that the deficit angle vanishes for a box with flat geometry. At the boundary vertices $A_v$ of (\ref{eq:av}) contains a sum of two triangle areas. In a different setting where the box has neighbor regions, the neighbor triangle areas would be included in the sum for $A_v$.

The numerical algorithm is as presented in \cref{sec:na}. For any fixed flow time $T$, we apply the adaptive Markov Chain Monte Carlo algorithm of \cite{Roberts2009ExamplesMCMC} to generate an ensemble of configurations according to the probability weight $e^{\Re E_{\text{eff}}}$. In each step we randomly pick an edge $e$, and propose a shift of $\sigma_e$ according to a Gaussian probability distribution. The variance of the distribution is dynamical in the adaptive MCMC algorithm employed here. In this algorithm, the acceptance rate is checked every $N$ ($N=50$ here) steps. If the acceptance rate is below or above the target rate $r=0.44$, the jump size is decreased or increased by 
\begin{align}
\delta(n)=\min(0.01,n^{-1/2}),
\end{align}
where $n$ is the step number. That $\delta(n)\rightarrow 0$ as $n\rightarrow \infty$ ensures the asymptotic convergence of the chain. 

A proposal is rejected if the Lorentzian triangle inequality is violated. In another model, one may also choose to reject a proposal if the number of light rays at a vertex is different from that of the flat configuration. However, in the symmetry-reduced box model the triangle inequality automatically implies the light ray number matching, so only the triangle inequality needs to be imposed. With this constraint, the dynamical edges can still be either timelike or spacelike. 

A lower bound $E_\text{min}=-10.0$ is imposed on $\Re E_{\text{eff}}$ in the numerical integration for the holomorphic flow from $t=0$ to the designated flow time $t=T$. If $\Re E_{\text{eff}}$ is too small the proposal will not be accepted. It improves the efficiency of the algorithm to simply truncate the integrator at the lower bound to move on to the next proposal.


\subsection{Results}

We consider five sets of coupling constants $p$. The numerical simulations are performed using the Julia programming language \cite{Bezanson2017Julia:Computing} on a personal computer. All Markov chains are obtained within about an hour. 
In all cases, we are able to identify a flow time $T$ for which the sign problem is significantly ameliorated so that $\Phi>0.9$.

\begin{figure}
    \centering
    \includegraphics[width=0.8\textwidth]{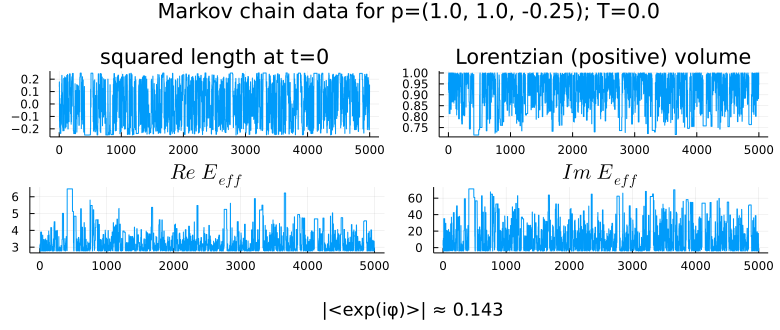}
    \caption{The starting case with $p=(1.0, 1.0, -0.25)$. With $T=0.0$ the phase fluctuation is quite large.}
    \label{fig:01-01}
\end{figure}

\begin{figure}
    \centering
    \includegraphics[width=0.8\textwidth]{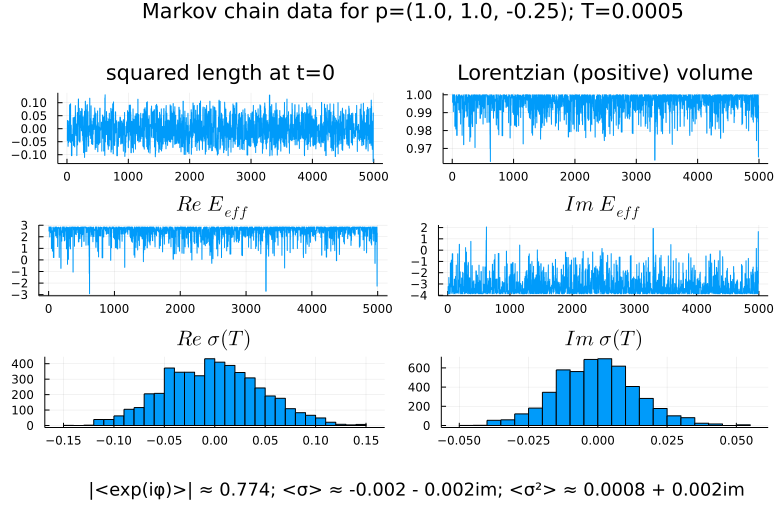}
    \caption{The starting case with $p=(1.0, 1.0, -0.25)$. With $T=0.0005$ the phase fluctuation is moderately suppressed.}
    \label{fig:01-02}
\end{figure}

\begin{figure}
    \centering
    \includegraphics[width=0.8\textwidth]{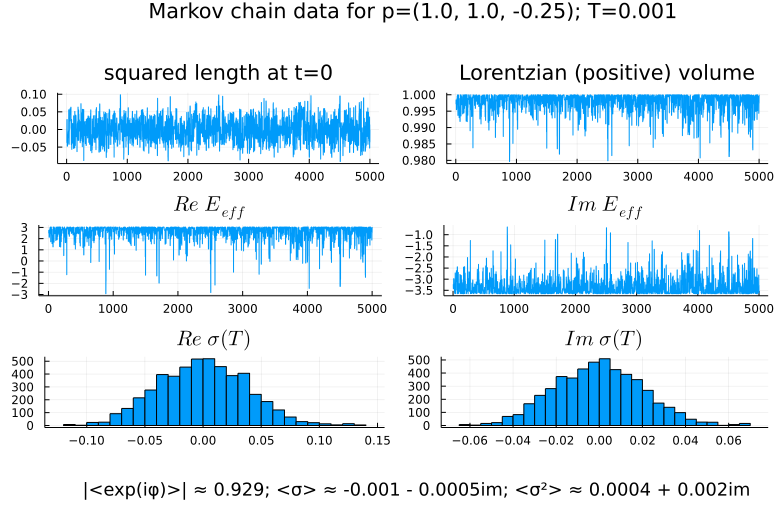}
    \caption{The starting case with $p=(1.0, 1.0, -0.25)$. With $T=0.001$ the phase fluctuation is moderately suppressed.}
    \label{fig:01-03}
\end{figure}

\subsection*{Starting case}

For $p=(1.0, 1.0, -0.25)$ where $m=-0.25$ for the DeWitt measure in $2D$ \cite{Hamber2009QuantumApproach}), we consider $T=0.0, T=0.0005$ and $T=0.001$ (\Cref{fig:01-01} to \Cref{fig:01-03}). As the flow time $T$ is increased from $0.0$ to $0.001$, the average phase $\Phi$ increases from about $0.143$ to $0.929$, which is close to $1$ and indicates that the phase fluctuation becomes much suppressed.

Note that $\ev{\sigma}\approx 0$, which is not a coincidence since the model admits a $\mathbb{Z}_2$ symmetry. One can check that the transformation $\sigma\mapsto-\sigma$ on the interior squared length preserves the path integral amplitude. Therefore for any $\sigma$ configuration there is the $-\sigma$ configuration with opposite contribution to $\ev{\sigma}$ to make $\ev{\sigma}=0$ as an exact result.

On the other hand, even though the numerical estimation of $\ev{\sigma^2}$ is close to zero, its value is not expected to vanish. That $\sigma^2$ is small is simply because it is the square of $\sigma$ which is close to zero. The third row in the figure with $T=0.001$ shows the histograms for the real and imaginary parts of $\sigma$ evaluated at the flow time $T$. The finite width of the distribution indicates the presence of fluctuations for the magnitude of $\sigma$. 

In the following, we will change the parameters one by one to see how this influences the fluctuations reflected in the histograms and the estimated values of $\ev{\sigma^2}$.

\begin{figure}
    \centering
    \includegraphics[width=0.8\textwidth]{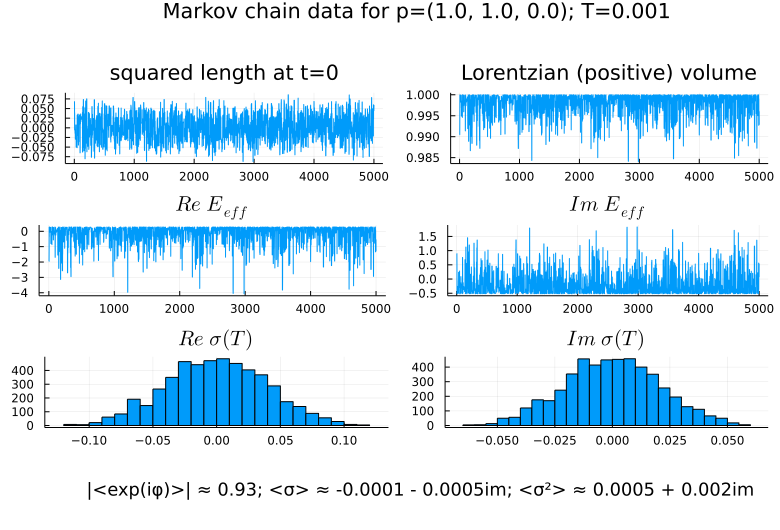}
    \caption{Increasing $m$ to $0.0$ does not influence the fluctuation in $\sigma$ much.}
    \label{fig:02-01}
\end{figure}

\subsection*{Changing m}

Given a new problem with a new set of parameters $p$, at present we do not know how to determine beforehand a suitable value of $T$ with small enough phase fluctuation. Therefore we simply find a suitable value of $T$ with $\Phi>0.9$ by trial and error. Here and below, we directly show the results for the suitable $T$.

The result for $m$ increased to $0.0$ is shown in \Cref{fig:02-01}. No significant difference is seen in the histogram or in $\ev{\sigma^2}$ in comparison to the original case of $m=-0.25$.

\begin{figure}
    \centering
    \includegraphics[width=0.8\textwidth]{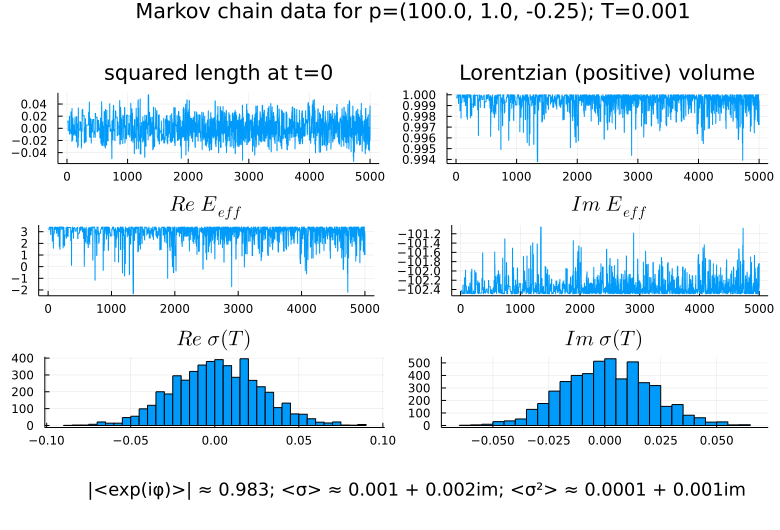}
    \caption{Changing $\lambda$ to $100.0$ slightly reduces the fluctuation in $\sigma$.}
    \label{fig:03-01}
\end{figure}

\begin{figure}
    \centering
    \includegraphics[width=0.8\textwidth]{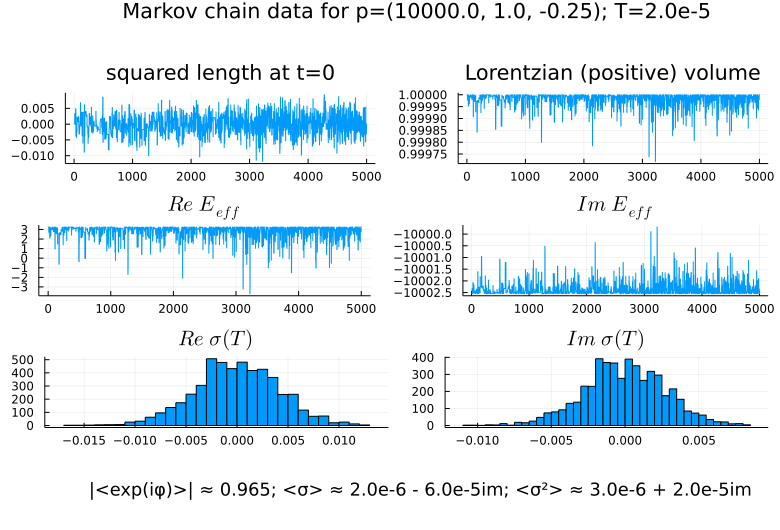}
    \caption{Changing $\lambda$ to $10000.0$ largely reduces the fluctuation in $\sigma$.}
    \label{fig:03-02}
\end{figure}

\subsection*{Changing $\lambda$}

The results for $\lambda$ changed to $100$ and $10000$ are shown in \Cref{fig:03-01} and \Cref{fig:03-02}. Although it may not be so apparent from just the cases of $\lambda=1$ and $\lambda=100$, including the case of $\lambda=10000$ makes it clear that the fluctuation in $\sigma$ in reduced, as indicated by the decreased width of the histogram distribution and the decreased magnitude of $\ev{\sigma^2}$.

\begin{figure}
    \centering
    \includegraphics[width=0.8\textwidth]{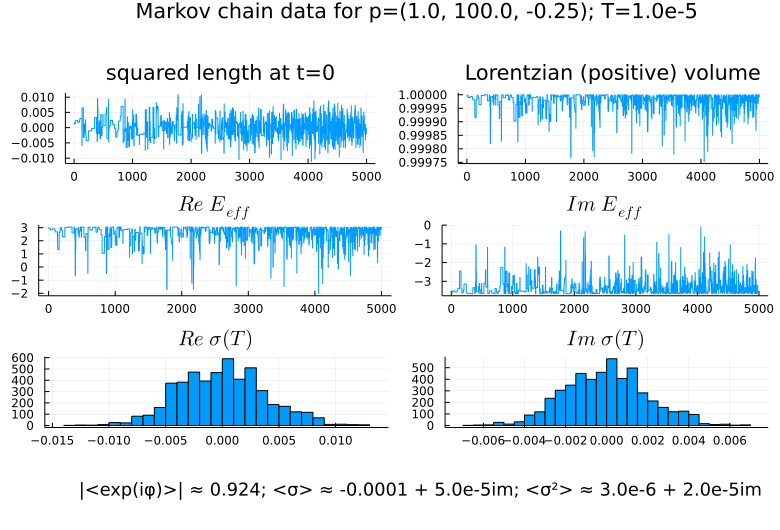}
    \caption{Changing $a$ to $100.0$ largely reduces the fluctuation in $\sigma$.}
    \label{fig:04-01}
\end{figure}

\subsection*{Changing $a$}

The results for $a$ changed to $100$ is shown in \Cref{fig:04-01}. In comparison to the cases of $a=1$, it is quite clear that increasing $a$ reduces the fluctuation in $\sigma$.

\subsection{Contour boundaries}\label{sec:cb}

As mentioned in \cref{sec:na}, to apply the holomorphic gradient flow algorithm we need that: 1) The holomorphic flow transverse a region where the path integrand is holomorphic; 2) The boundary of the flow region have negligible contribution to the path integral.

For simplicial quantum gravity, the boundaries are set by the branch point singularities of the path integrand, the generalized triangle inequalities, and additional constraints such as the light ray number constraint mentioned above. Within the region bounded by these boundaries, the path integrand is holomorphic, so the first requirement is met. 

We now check the second requirement that the boundary of the flow region make negligible contribution to the path integral. We noted above that for the symmetry-reduced box model, the generalized triangle inequalities imply the light ray number constraint. In addition, the boundaries of the generalized triangle inequalities are set where the Lorentzian volumes vanish, i.e., $\sV_t=0$. Yet this coincides with one of the square root branch points singularities (see (\ref{eq:2dvol}) and (\ref{eq:ca})). Therefore altogether we only need to consider the boundaries of the branch point singularities of the path integrand. 

Along such boundaries the contribution to the path integral is infinitely suppressed. To see this, note from (\ref{eq:DreEDt}) that $\dv{E_R}{t}=-\sum_e\abs{\partial_e E}^2\le 0$, i.e., the real part of the path exponent $E$ decays monotonically at a rate determined by $\abs{\partial_e E}$ along the holomorphic flow. Using the formulas of \cref{sec:2dsqg}, one can check that $\abs{\partial_e E}\rightarrow\infty$ at the branch point singularities. Therefore at the boundaries set by these branch points, the path integrand is infinitely exponentially suppressed. They make negligible contributions to the path integral. 

\section{Discussion}\label{sec:d}

We have provided a definition of complex simplicial gravity, which reduces to Euclidean and Lorentzian simplicial gravity in special cases. 

The complex formalism enabled us to perform Monte Carlo simulations for Lorentzian simplicial quantum gravity. The numerical sign problem is overcome by deforming the integration contour into the complex.

The complex formalism also sets the path for further studies of singularity resolving processes with complex semi-classical solutions, generalizing previous studies in the symmetry-reduced setting \cite{Hartle1989SimplicalModel, Louko1992ReggeCosmology, Birmingham1995LensCosmology, Birmingham1998ACalculus, Furihata1996No-boundaryUniverse, Silva1999SimplicialField, Silva1999AnisotropicField, Silva2000SimplicialPhi2, daSilvaWormholesMinisuperspace, Dittrich2022LorentzianSimplicial}, and making a clear connection to the Lorentzian theory.

The numerical simulations for Lorentzian simplicial quantum gravity performed here are in a very simple setting. They are on a simple box lattice, in $1+1D$, with symmetry reduction, and for pure gravity. Future works should extend to larger lattices, higher dimensions, without symmetry reduction, and with matter coupling. 

The physics theory side of these generalizations is understood. From the present work it is clear how to define complex simplicial quantum gravity on larger lattices in higher dimensions without symmetry reduction. From previous works it is clear how to couple to the matter species of the Standard Model (see e.g., Chapter 6 of Hamber's textbook \cite{Hamber2009QuantumApproach} and references therein). 

The numerics side of these generalizations still needs to be understood better. It is unclear to what extent the holomorphic gradient flow algorithm adopted here will remain efficient. Some other techniques may be needed, such as the tempered thimbles, the learnifolds, and the path optimization algorithms reviewed in \cite{AlexandruComplexProblem} and further developed in, e.g., \cite{Fukuma2021WorldvolumeMethod, Fukuma2021StatisticalAlgorithm, Lawrence2021NormalizingProblem, Wynen2021MachineProblems}. 

Using the numerical tools, one could study the refinement (continuum) limit of the theory. One could investigate questions about the fate of black hole and cosmological singularities (see the Introduction section for a list of references on this topic). From a path integral perspective, if a process can be characterized by a set of path integral configurations, the formalism assigns a probability to it (which may or may not have meaning to cognitive beings such as us). 
Simplicial quantum gravity provides a formalism to compute and compare the probabilities for such processes.

\section*{Acknowledgement}

I am very grateful to Bianca Dittrich and Jos{\'e} Padua-Arg{\"u}elles for comments and questions that helped improve the manuscript, to Seth Asante, Lee Smolin and Bianca Dittrich for valuable discussions on simplicial quantum gravity, to Erik Schnetter and Dustin Lang for timely help on computation matters, and to Lucien Hardy, Achim Kempf, Laurent Freidel, and Robert Mann for valuable discussions on quantum gravity in general.

Research at Perimeter Institute is supported in part by the Government of Canada through the Department of Innovation, Science and Economic Development Canada and by the Province of Ontario through the Ministry of Economic Development, Job Creation and Trade. 

\bibliographystyle{unsrt}
\bibliography{mendeley.bib}

\end{document}